\documentclass[a4paper,english]{article} 

\usepackage[left=3.2cm, right=3.8cm]{geometry}

\usepackage{microtype}
\usepackage{authblk}
\usepackage{amsmath}

\usepackage{amsthm}
\theoremstyle{plain}
\newtheorem{theorem}{Theorem}

\usepackage[ruled,vlined]{algorithm2e}

\usepackage{url}
\usepackage{subcaption}
\usepackage{lscape}
\usepackage{longtable}
\usepackage{booktabs}
\usepackage{listings}
\usepackage{tabu}
\usepackage{pgfplots}
\pgfplotsset{compat=1.12}
\newcolumntype{R}{>{\raggedleft\arraybackslash}p{1.15cm}}

\graphicspath{{./figures/}{./plots/}}

\title{Solving Large-Scale Minimum-Weight Triangulation Instances to Provable Optimality}

\author[1]{Andreas Haas}
\affil[1]{Department of Computer Science, TU Braunschweig, 38106 Braunschweig, Germany.
  \texttt{haas@ibr.cs.tu-bs.de}}

\DeclareMathOperator{\MWT}{MWT}

\DeclareMathOperator{\dist}{d}
\DeclareMathOperator{\sign}{sign}

\begin{document}

\maketitle

\begin{abstract}
We consider practical methods for the problem of finding a minimum-weight triangulation (MWT) of a planar point set, a classic problem of computational geometry with many applications. While Mulzer and Rote proved in 2006 that computing an MWT is NP-hard, Beirouti and Snoeyink showed in 1998 that computing provably optimal solutions for MWT instances of up to 80,000 {\em uniformly distributed} points is possible, making use of clever heuristics that are based on geometric insights.
We show that these techniques can be refined and extended to instances of much bigger size and different type, based on an array of modifications and parallelizations in combination with more efficient geometric encodings and 
data structures. As a result, we are able to solve MWT instances with up to 30,000,000 uniformly distributed points in less than 4 minutes to provable optimality. Moreover, we can compute optimal solutions for a vast array of other 
benchmark instances that are {\em not} uniformly distributed, including normally distributed instances (up to 30,000,000 points), all point sets in the TSPLIB (up to 85,900 points), and VLSI instances with up to 744,710 points. This demonstrates that from a practical point of view, MWT instances can be handled quite well, despite their theoretical difficulty.
\end{abstract}
\section{Introduction}

Triangulating a set of points in the plane is a classic problem in computational geometry: given a planar point set $S$, find a maximal set of non-crossing line segments connecting the points in $S$. Triangulations have many real-world applications, for example in terrain modeling, finite element mesh generation and visualization.
In general, a point set has exponentially many possible triangulations and a natural question is to ask for a triangulation that is optimal with respect to some optimality criterion. Well known and commonly used is the Delaunay triangulation, which optimizes several criteria at the same time: it maximizes the minimum angle and minimizes both the maximum circumcircle and the maximum smallest enclosing circle of all triangles.
Another natural optimality criterion, and the one we are considering in this paper is minimizing the total weight of the resulting triangulation, i.e., minimizing the sum of the edge lengths.

The minimum-weight triangulation ($\MWT$) is listed as one of the open problems in the famous book from 1979 by Garey and Johnson on NP-completeness \cite{garey1979}.
The complexity status remained open for 27 years until Mulzer and Rote \cite{mulzer2006} finally resolved the question and showed NP-hardness of the $\MWT$ problem.

Independently, Gilbert \cite{gilbert1979} and Klincsek \cite{klincsek1980} showed that, when restricting it to simple polygons, the $\MWT$ problem can be solved in $O(n^3)$-time with dynamic programming.
The dynamic programming approach can be generalized to polygons with $k$ inner points. Hoffmann and Okamoto \cite{hoffmann2004} showed how to obtain the $\MWT$ of such a point set in $O(6^k n^5 \log n)$-time. Their algorithm is based on a polygon decomposition through $x$-monotone paths.
Grantson et al. \cite{grantson05} improved the algorithm to $O(n^4 4^k k)$  and showed another decomposition strategy based on cutting out triangles \cite{grantson05triangles} which runs in $O(n^3 k! k)$-time.

A promising approach are polynomial-time heuristics that either include or exclude edges with certain properties from any minimum-weight triangulation.
Das and Joseph \cite{das1989diamondproperty} showed that every edge in a minimum-weight triangulation has the \emph{diamond property}. An edge $e$ cannot be in $\MWT(S)$ if both of the two isosceles triangles with base $e$ and base angle $\pi/8$ contain other points of $S$.
Drysdale et al. \cite{drysdale2001exclusionregion} improved the angle to $\pi/4.6$.
This property can exclude large portions of the edge set and works exceedingly well on uniformly distributed point sets, for which only an expected number of $O(n)$ edges remain.
Dickerson et al.  \cite{dickerson96, dickerson97} proposed the \emph{LMT-skeleton heuristic}, which is based on a simple local-minimality criterion fulfilled by every edge in $\MWT(S)$. The LMT-skeleton algorithm often yields a connected graph, such that the remaining polygonal faces can be triangulated with dynamic programming to obtain the minimum weight triangulation.

Especially the combination of the diamond property and the LMT-skeleton made it possible to compute the $\MWT$ for large, well-behaved point sets. Beirouti and Snoeyink \cite{snoeyink98implementations,beirouti1997thesis} showed an efficient implementation of these two heuristics and they reported that their implementation could compute the exact $\MWT$ of 40,000 uniformly distributed points in less than 5 minutes and even up to 80,000 points with the improved diamond property.

Our contributions:
\begin{itemize}
\item We revisit the diamond test and LMT-skeleton based on Beirouti's and Snoeyink's \cite{snoeyink98implementations,beirouti1997thesis} ideas and describe several improvements. Our bucketing scheme for the diamond test does not rely on a uniform point distribution and filters more edges. 
For the LMT-skeleton heuristic we provide a number of algorithm engineering modifications. They contain a data partitioning scheme for a parallelized implementation and several other changes for efficiency. 
We also incorporated an improvement to the LMT-skeleton suggested by Aichholzer et al. \cite{aichholzer99}.

\item We implemented, streamlined and evaluated our implementation on various point sets. For the uniform case, we computed the $\MWT$ of 30,000,000 points in less than 4 minutes on commodity hardware; the limiting factor arose from the memory of a standard machine, not from the runtime. We achieved the same performance for normally distributed point sets. 
The third class of point sets were benchmark instances from the TSPLIB \cite{reinelt1991} (based on a wide range of real-world and clustered instances) and the VLSI library. These reached a size up to 744,710 points.
This shows that from a practical point of view, wide range of huge $\MWT$ instances can be solved to provable optimality with the right combination of theoretical insight and algorithm engineering.
\end{itemize}

\section{Preliminaries}
Let $S$ be a set of points in the euclidean plane. A triangulation $T$ of $S$ is a \emph{maximal planar} straight-line graph with vertex set $S$.  The weight $w(e)$ of an edge $e$ is its euclidean length.
A minimum-weight triangulation $\MWT(S)$ minimizes the total edge weight, i.e., $\sum_{e \in E(T)} w(e)$.

An edge $e$ is \emph{locally minimal} with respect to some triangulation $T(S)$ if either
\begin{enumerate}
\item[(i)] $e$ is an edge on the convex hull of $S$, or
\item[(ii)] the two triangles bordering $e$ in $T(S)$ form a quadrilateral $q$ such that $q$ is either
not convex or $e$ is the shorter diagonal of the two diagonals $e$ and $e'$ in $q$, i.e., $w(e) \le w(e')$.
\end{enumerate}

A triangulation $T$ is said to be a \emph{locally minimal triangulation} if every edge of $T$ is locally minimal, i.e., the weight of $T$ cannot be improved by edge flipping.
A pair of triangles witnessing local minimality for some edge $e$ in some triangulation is called a \emph{certificate} for $e$.
An \emph{empty triangle} is a triangle that contains no other points of $S$ except for its three vertices.
\section{Previous Tools}
\label{sec::previous_tools}

\subsection{Diamond Property}
\label{ssec::prev::diamond_property}


\begin{figure}[h]
\centering
\includegraphics[width=0.5\columnwidth]{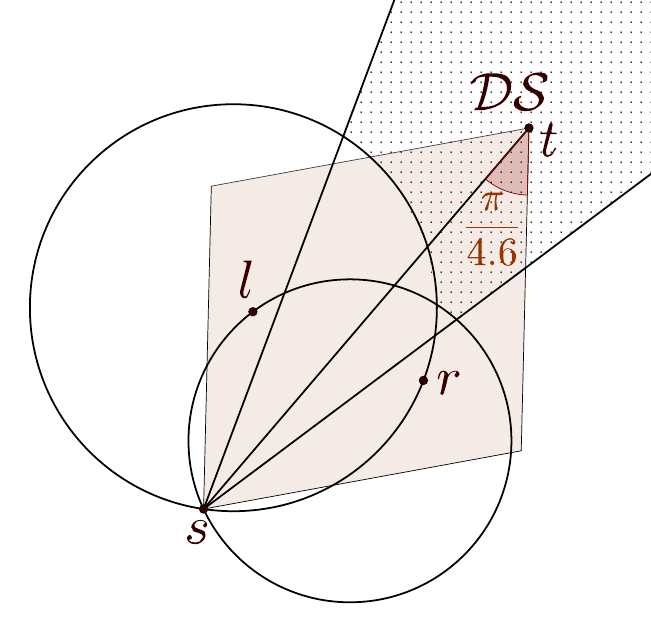}
\caption[Dead sector $\mathcal{DS}$]{Points $l$ and $r$ induce a region $\mathcal{DS}$ such that all edges $e=st$ with $t \in \mathcal{DS}$ fail the diamond test. $\mathcal{DS}$ is called a dead sector (dotted area).}
\label{fig::dead_sector}
\end{figure}

A brute-force solution to test the diamond property for each edge takes $\Theta(n^3)$ time and is inefficient. To accelerate the test, Beirouti and Snoeyink \cite{snoeyink98implementations} use a bucketing scheme based on a uniform grid with the grid size chosen such that on expectation a constant number of points lie in each cell. In order to quickly discard whole cells, they make use of dead sectors, which are illustrated in Figure~\ref{fig::dead_sector}. Suppose we want to test all edges with source $s$, points $l,r$ are already processed and known, then all edges $st$ with $t \in \mathcal{DS}$ will fail the diamond test, because $l$ and $r$ lie in the left, resp.~right isosceles triangle. The boundary of a single dead sector depends on the angle and length of edges $sl$ and $sr$; for multiple sectors it can be quite complicated.
For each point $s$, cells are searched starting at the cell containing $s$ until all cells can be excluded by dead sectors.

\subsection{LMT-Skeleton}
\label{ssec::prev::lmt_skeleton}

The LMT-skeleton was proposed by Dickerson et al. \cite{dickerson96, dickerson97}, it is a subset of the minimum weight triangulation. The key observation is that $\MWT(S)$ is a locally minimal triangulation, i.e., no edge in $\MWT(S)$ can be flipped to reduce the total weight. 

\begin{figure}
\begin{subfigure}{0.54\columnwidth}
\includegraphics[width=0.95\columnwidth]{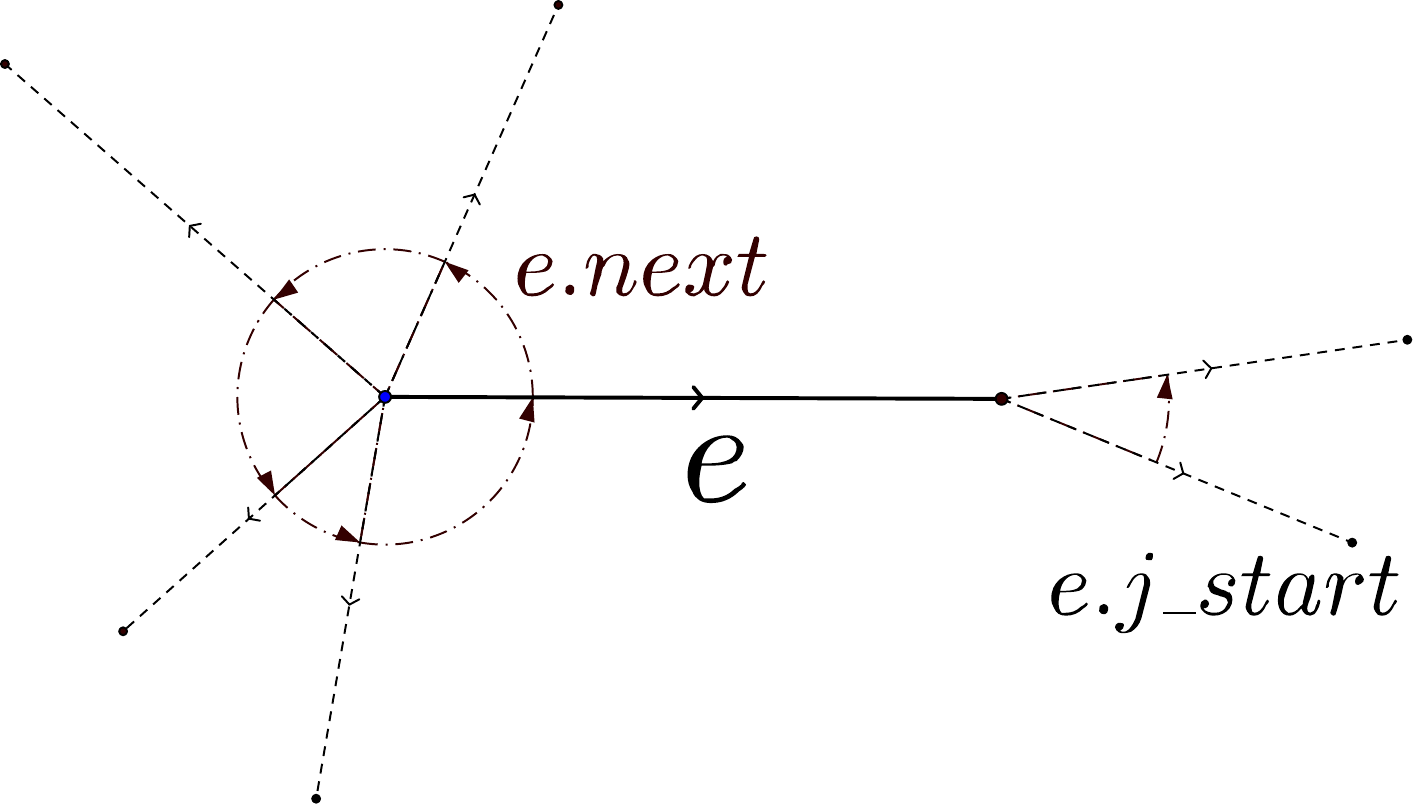}
\caption{Half-edge $e$.}
\label{fig::lmt_halfedge_a}
\end{subfigure}
\begin{subfigure}{0.44\columnwidth}
\includegraphics[width=0.95\columnwidth]{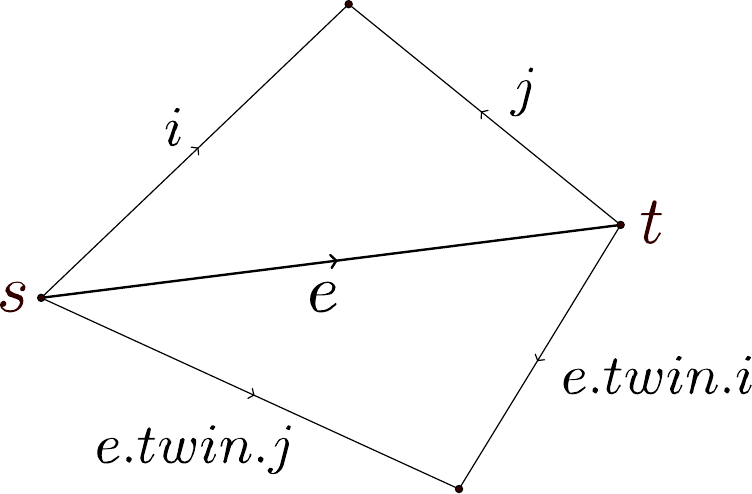}
\caption{$i$ and $j$ point to the current certificate edges.}
\label{fig::lmt_halfedge:ij}
\end{subfigure}
\caption{Representation of half-edge $e$.}
\label{fig::lmt_halfedge}
\end{figure}

In order to avoid the $O(n^3)$ space required to store all empty triangles Beirouti and Snoeyink \cite{snoeyink98implementations} propose a data structure based on half-edges. Half-edges store several pointers: a pointer to the target vertex and the twin half-edge; a pointer to the next edge in counter-clockwise order around the source; and three additional pointers to scan for empty triangles ($i$, $j$, $j\_start$); see Figure \ref{fig::lmt_halfedge} for an illustration. A status flag indicates whether an edge is \emph{possible, impossible} or \emph{certain}.
Furthermore, three additional pointers (\emph{rightPoly, leftPoly, polyWeight}) are stored and used for the subsequent polygon triangulation step.

\begin{algorithm}
\SetKwFunction{Adv}{Advance}

\SetKwProg{function}{Function}{}{}

\function{\Adv{$e$}}{
	\Repeat{$e.i.target = e.j.target$}{
		\While{$e.i.target$ $\textsf{\upshape is not left of}$ $e.j$}{
			$e.i \gets e.i.next$ \;
		}
		\While{$e.i.target$ $\textsf{\upshape is left of}$ $e.j$}{
			$e.j \gets e.j.next$ \;
		}
	}
}
\caption{Advance. Adapted from \cite{beirouti1997thesis} (Changed notation and corrected an error.)}
\label{alg::advance}
\end{algorithm}

At the heart of the LMT-skeleton heuristic lies the \Adv function, see Algorithm~\ref{alg::advance}.
\Adv basically rotates edge $i$ and $j$ in counterclockwise order such that they form an empty triangle if they point to the same vertex, i.e., $i.target = j.target$. Pointers $i$ and $j$ are initialized to $e.next$ resp. $e.j\_start$.
The algorithm to find certificates is built on top of \Adv. All pairs of triangles can be traversed by repeatedly calling \Adv on half-edge $e$ and $e$'s twin in fashion similar to a nested loop. The ``loop'' is stopped when a certificate is found and can be resumed when the certificate becomes invalid.  \cite{beirouti1997thesis, snoeyink98implementations}

After initializing the half-edge data structure, their implementation pushes all edges on a stack (sorted with longest edge on top) and then processes edges in that order.
If for an edge $e$ no certificate is found, an intersection test determines if $e$ lies on the convex hull or if $e$ is impossible.
If $e$ is detected to be impossible, a local scan restacks all edges with $e$ in their certificate. After the stack is empty, all edges that remain possible and that have no crossing edges are marked as \emph{certain}.

\section{Our Improvements and Optimizations}
\subsection{Diamond Property}
\label{sec::diamond_property}

For a uniformly distributed point set $S$ with $|S|=n$ points, the expected number of edges to pass the diamond test is only $O(n)$. More precisely, Beirouti and Snoeyink \cite{snoeyink98implementations} state that the number is less than $3\pi n / \sin(\alpha)$, where $\alpha$ is the base angle for the diamond property. We were able to tighten this value.

\begin{theorem}
Let $S$ be a uniformly distributed point set in the plane with $|S| = n$ and let $\alpha \le \pi/3$ be the base angle for the diamond property. Then the expected number of edges that pass the diamond test is less than $3\pi n / \tan(\alpha)$.
\end{theorem}

\begin{proof}
Fix an arbitrary point $s \in S$ and consider the remaining points $t_i$ , $0 \le i \le n-2$ in order of increasing distance to $s$. Edge $e_i := st_i $ fulfills the diamond property if at least one of the two corresponding isosceles triangles is empty, i.e., it contains non of the $i$ points  $t_0, \dots ,t_{i-1}$. 

For any given distance $r$, each triangle has area $A = 1/4 \tan(\alpha) r^2$. The points are uniformly distributed in the circle centered at $s$ with radius $r$, thus the probability $p$ that a fixed point lies in a fixed triangle is $p = A/\pi r^2 = \tan(\alpha)/4\pi$. Each triangle is empty with probability $(1-p)^i$. The whole diamond is empty with probability $(1-2p)^i $. It follows that at least one of the triangles is empty with probability $2(1-p)^i - (1-2p)^i$. 

Let $X_i$ be the indicator variable of the event that edge $e_i$ fulfills the diamond property. Then $X = \sum_{0}^{n-2} X_i$ is the number of outgoing edges that pass the diamond test.
By linearity of expectation and the geometric series, the expected value of $X$ is bounded by

\begin{equation*}
E[X] = \sum_{i=0}^{n-2} E[X_i] < \sum_{i=0}^{\infty} 2(1-p)^i - (1-2p)^i = \frac{2}{p} - \frac{1}{2p} = \frac{3}{2p} = \frac{6\pi}{\tan(\alpha)}
\end{equation*}

If we apply the same argument to each point in $S$, we are counting each edge twice. Hence the number of edges that pass the diamond test with base angle $a$ is less than $3\pi n / \tan(\alpha)$. 
\end{proof}

For $\alpha = \pi/4.6$ we get a value less than $11.5847$, which is very close to the values observed and achieved by our implementation; see Table~\ref{tbl::diamond_uniform_1} in Section \ref{sec::experiments}. In contrast, the value achieved by the implementation of Beirouti and Snoeyink is $\approx 14.3$ \cite{snoeyink98implementations}. 

\subsection{Dead Sectors and Bucketing}

Our bucketing scheme is based on the same idea of dead sectors as described by Snoeyink and Beirouti \cite{snoeyink98implementations}. Our implementation differs in two points. Despite being simpler, it has higher accuracy and it can easily be integrated into common spatial data structures; such as quadtrees, kd-trees and R-trees. Therefore, it is not limited to uniformly distributed point sets.

\begin{figure}
\centering
\includegraphics[width=0.40\columnwidth]{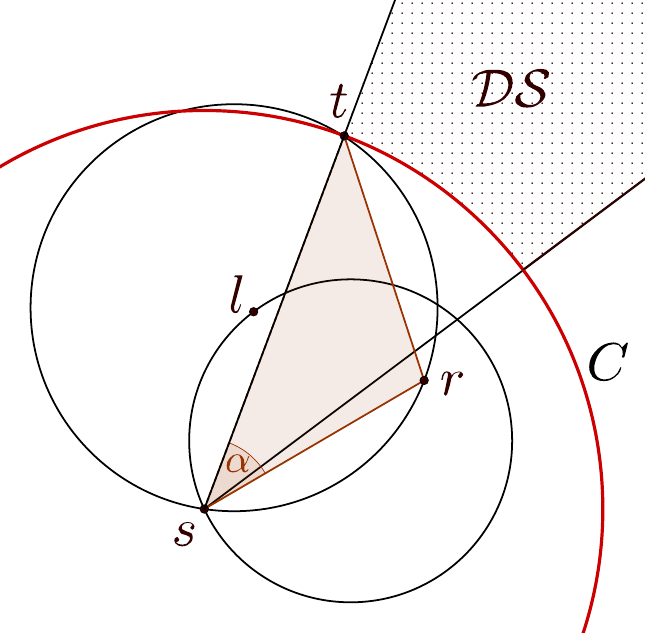}
\caption[The simplified dead sector $\mathcal{DS}$]{Simplified dead sector $\mathcal{DS}$ is bounded by two rays and circle $C$.
$\mathcal{DS}$ is stored as a triple: an interval of two polar angles and a squared radius.}
\label{fig::simple_dead_sector}
\end{figure}

In order to avoid storing complicated sector boundaries, we simplify the shape. Instead of bounding a sector $\mathcal{DS}$ by two circles as illustrated in Figure~\ref{fig::dead_sector}, we only use a single big circle $C$ with center $s$ at the expense of losing a small part of $\mathcal{DS}$. This allows a compact representation of dead sectors as a triple of three numbers: an interval consisting of two polar angles and a squared radius; see Figure~\ref{fig::simple_dead_sector}.

The main ingredient for our bucketing scheme is a spatial search tree with support for incremental nearest neighbor searches, such as quadtrees, kd-trees or R-trees. A spatial search tree hierarchically subdivides the point set into progressively finer bounding boxes/rectangles until a predefined threshold is met. Incremental nearest neighbor search queries allow to traverse all nearest neighbors of a point in order of increasing distance. Such queries can easily be implemented by utilizing a priority queue that stores all tree nodes encountered during tree traversal together with the distances to their resp.~bounding box (see Hjaltason and Samet \cite{hjaltason1995}). 

Pruning tree nodes whose bounding box lie in dead sectors is rather simple as follows:
consider a nearest neighbor query for point $s$: when we are about to push a new node $n$ into the priority queue, we compute the smallest polar angle interval $I$ that encloses the bounding box of $n$ and discard $n$ if $I$ is contained in the dead sectors computed so far.
The interval of a bounding box is induced by the two extreme corners as seen from $s$, i.e., the leftmost and the rightmost corner. 

Because nearest neighbors and tree nodes are processed in order of increasing distance, we can store sectors in two stages. On creation, they are inserted into a FIFO-queue; later only the interval component is inserted in a search filter used by the tree. The queue can be seen as a set of pending dead sectors with an attached activation distance $\delta$. As soon as we process a point $t$ with $\dist(s,t) > \delta$ we can insert the corresponding interval into our filter. 

This reduces the data structure used for the filter to a simple set of sorted non-overlapping intervals consisting of polar angles. Overlapping intervals are merged on insertion, which reduces the maximal number of intervals that need to be tracked at the same time to a very small constant \footnote{The exact value is 15 in our case, but it depends on an additional parameter and implementation details.}.

This leaves the issue of deciding which points are used to construct dead sectors. We store all points encountered during an incremental search query in an ordered set $N$ sorted by their polar angle with respect to $s$. Every time we find a new point $t$, it is inserted into $N$ and dead sectors are computed with the predecessor and the successor of $t$ in $N$. There is no need to construct sectors with more than the direct predecessor and successor, because sectors between all adjacent pairs of points in $N$ were already constructed on earlier insertions. Computing the activation distance for new sectors only requires a single multiplication of the current squared distance to $t$ with a precomputed constant. Additionally, the diamond property of edge $st$ is tested against a subset of $N$. 

If we apply the above procedure to every single point, we generate each edge twice, once on each of the two endpoints. Therefore, we output only those edges $e=st$ such that $s < t$, i.e., $s$ is lexicographically smaller than $t$. As a consequence, we can exclude the left half-space right from the beginning by inserting an initial dead sector $\mathcal{DS}_0 = (1/2 \pi, 3/2 \pi \rbrack$ at distance 0.

In order to increase cache efficiency we store the point set in a spatially sorted array. The points are ordered along a Hilbert curve, but the choice of a particular space-filling curve is rather arbitrary.  Our spatial tree implementation is a quadtree that is built on top of that array during the sorting step.
Profiling suggests the memory layout of the tree nodes is not important. We apply the diamond test to every single point and we can freely choose the order in which we process them. The points are spatially sorted and processed in this order, which leads to similar consecutive search paths in the tree and therefore most nodes are already in the CPU cache. 

In order to avoid the expensive transcendental \texttt{atan2} function for polar angle computations, we can use any function that is monotonic in the polar angle for comparisons between different angles. One such function, termed \emph{pseudo-angle}, was described by Moret and Shapiro \cite{moret1991PtoNP}. The basic idea is to measure arc lengths on the $L_1$ unit circle, instead of the $L_2$ unit circle. With some additional transformations, the function can be rewritten to $\sign(y)(1 - x /(|x| + |y|))$, where we define $\sign(0) =: 1$. This function has the same general structure as \texttt{atan2}: a monotonic increase in the intervals $[0, \pi]$, $(\pi, 2\pi)$ and a discontinuity at $\pi$, with a jump from positive to negative values. Additionally, it gives rise to a one-line implementation (see Figure~\ref{code::pseudo_angle} in Appendix \ref{app::pseudo_angles}), which gets compiled to branch-free code.

\subsection{LMT-Skeleton}
\label{ssec::lmt_skeleton}

For ``nicely'' distributed point sets, a limiting factor of the heuristic is the space required to store the half-edge data structure in memory.
In order to save some space we removed three variables from the original description (\emph{rightPoly, leftPoly, polyWeight}). They serve no purpose until after the heuristic, when they are used for the polygon triangulation step  (therefore, reducing cache-efficiency and wasting space). For edges marked \emph{impossible} (typically the majority), they are never used at all; for the remaining edges they can be stored separately as soon as needed. 
We further reduce storage overhead by storing all edges in a single array sorted by source vertex (also known as a \emph{compressed sparse row graph}). All outgoing edges of a single vertex are still radially sorted.  In addition to the statuses \emph{possible}, \emph{certain}, \emph{impossible}, we store whether an edge lies on the convex hull.

As mentioned in Section \ref{sec::previous_tools}, certificates are found by utilizing \Adv in fashion of a nested loop.
It is crucial to define one half-edge of each pair as the primary one to distinguish which half-edge corresponds to the outer resp.~inner ``loop''. The choice is arbitrary as long as it is consistent throughout the execution of the algorithm.

Another problem that went unnoticed emerges when the diamond test and LMT-skeleton are combined. In this case \Adv does not guarantee to find empty triangles; it may stop with non-empty triangles due to missing incident edges. An example is shown in Figure~\ref{fig::advance_fail}, where all edges with the exception of $f$ pass the diamond test; calling \Adv on $e$ yields a non-empty triangle.

\begin{figure}[h]
\centering
\includegraphics[width=0.6\columnwidth]{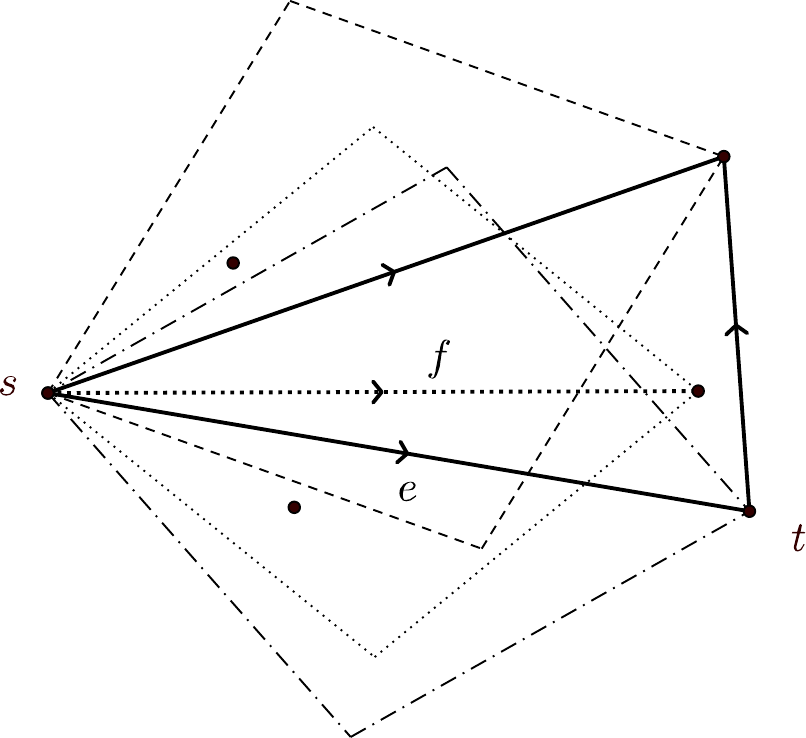}
\caption[Pitfall of the \texttt{Advance} function.]{Edge $f$ does not have the diamond property and in turn the \Adv function fails: it stops with a non-empty triangle.}
\label{fig::advance_fail}
\end{figure}

Fortunately, the side effect of wrong certificates is rather harmless. In the worst-case an otherwise impossible edge stays possible, which in turn may prevent other edges from being marked \emph{certain}, however, no edge will incorrectly be marked \emph{certain}. Even though invalid certificates occur frequently, we observed them to be of transient nature because some certificate edge itself becomes impossible later in the heuristic.
Therefore, we still use \Adv in our implementation to find certificates. However, it is important to keep in mind that the function can fail. Beirouti \cite{beirouti1997thesis} states that they also use \Adv to scan for empty triangles in simple polygons during the dynamic programming step after the LMT-skeleton. Even then \Adv can fail by returning triangles that are part of two adjacent simple polygonal faces; see Figure~\ref{fig::advance_fail_polygon} for an example.

\begin{figure}
\centering
\includegraphics[width=0.6\columnwidth]{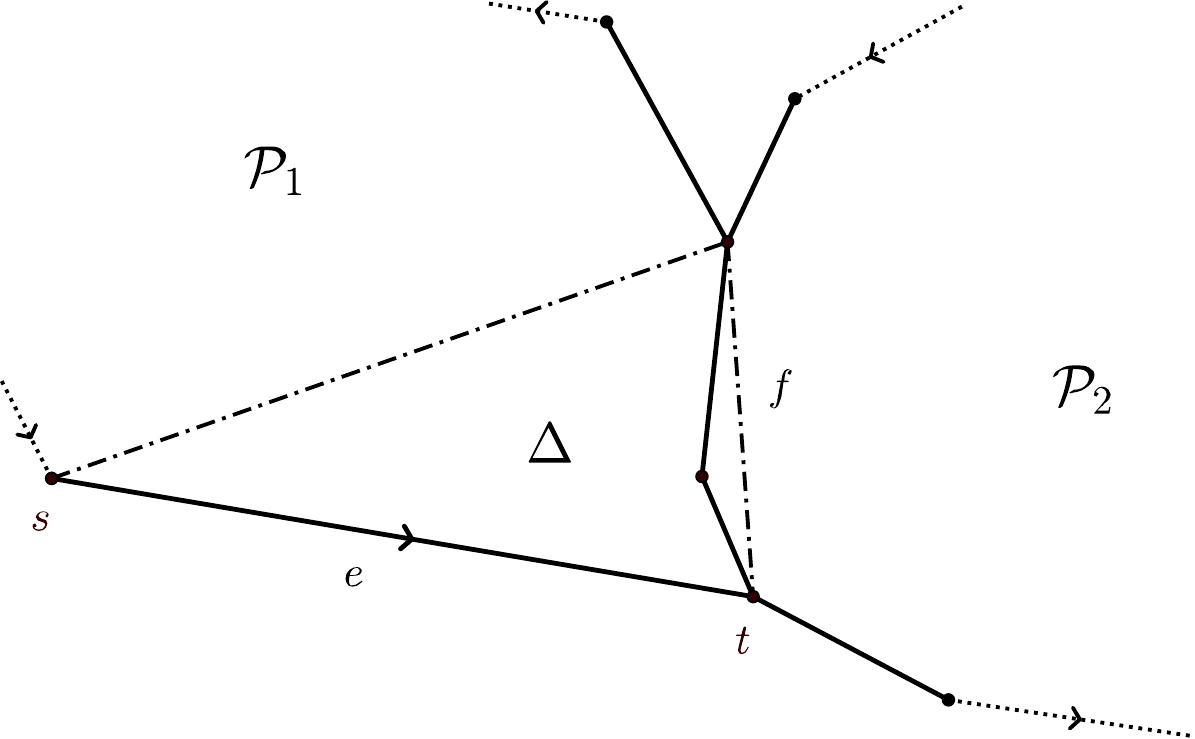}
\caption{\texttt{Advance} can return non-empty triangle $\Delta$ which is part of two adjacent polygonal faces.}
\label{fig::advance_fail_polygon}
\end{figure}

\begin{algorithm}
 \SetKwInOut{Input}{Input}
 \SetKwInOut{Output}{Output}
 \SetKwFunction{EmptyEdges}{EmptyEdges}
 \SetKwFunction{Restack}{RestackEdges}
 \SetKwFunction{Pop}{Pop}
 \SetKwFunction{LmtLoop}{LMT-Loop}
 \SetKwFunction{CH}{CH}
 \SetKwFunction{Intersect}{HasIntersections}
 
 \SetKwProg{function}{Function}{}{}

 
 \BlankLine
 \Begin{
	 Init data structures \;
	 $ST \gets$ \EmptyEdges{$S$}$\setminus$ \CH($S$) \tcc*{Stack $ST$}
	 
	 \LmtLoop($ST$) \;

	 \ForEach{\textsf{\upshape Possible primary half-edge} $e$}{
		\If{$\lnot$\Intersect($e$)}{
			Mark $e$ as certain\;
			}
	 } 
 }
 
 \function{\LmtLoop{$ST$}}{
	 \While{$ST$ \textsf{\upshape is not empty}}{
	 	$e \gets$ \Pop($ST$) \;
	 	\If{$e$ \textsf{\upshape has no certificate}}{
		 	\Restack(e) \tcc*{Push edges with $e$ in their certificate.}
		 	Mark $e$ as impossible\;
	 	}
	 } 	
 }
 
\caption{Refactored LMT-skeleton algorithm.}
\label{alg::lmt_skeleton}
\end{algorithm}

Pseudocode for our implementation is given in Algorithm~\ref{alg::lmt_skeleton}. In essence it is still the same as given by Beirouti and Snoeyink \cite{snoeyink98implementations}, however, with some optimizations applied.
First, the convex hull edges are implicitly given during initialization of the $j\_start$-pointers and can be marked as such without any additional cost.
Determining the convex hull edges beforehand allows to remove the case distinction inside the \LmtLoop, i.e., it removes all intersection tests that are applied to impossible edges.
Secondly, sorting the stack by edge length destroys spatial ordering and the loss of locality of reference outweighs all gains on modern hardware. 
Without sorting, it is actually not necessary to push all edges onto the stack upfront. 
Lastly, with proper partitioning of the edges, the \LmtLoop can be executed in parallel -- described in more detail in Section \ref{ssec::parallelization}.

Additionally, we incorporated an improvement to the LMT-skeleton suggested by Aichholzer et al.~\cite{aichholzer99}.
Consider a certificate for an edge $e$, i.e., a quadrilateral $q_e$ such that $e$ is locally minimal w.r.t.~$q_e$. It is only required that the four certificate edges $f_i \in q_e$ are not \emph{impossible}, that is, edge $f_i$ is either on the convex hull or in turn has some certificate $q_i$. Notice that $q_i$ and $q_e$ may not share a common triangle. However, if for edge $f_i$ there is no such certificate $q_i$ that shares a triangle with $q_e$, then edge $e$ cannot be in any locally minimal triangulation and $e$ can be marked \emph{impossible}.

The improved LMT-skeleton is computationally much more expensive.
Consider the case in which edge $e=(s,t)$ becomes impossible. In order to find invalid certificates, it is no longer sufficient to scan only those edges incident to either $s$ or $t$. In addition to edges of the form $(s,u)$, resp.~$(t,u)$, we also have to check all edges incident to any adjacent vertex $u$ for invalid certificates. Because edges do not store the certificates for their certificate it gets even worse: we cannot know if an edge has to be restacked and we must restack and recheck all of them. 
Another consequence is that we cannot resume the traversal of triangles for any edge $f_i$, because we do not know where we stopped the last time.

We are left with a classic space-time trade-off and we chose not to store any additional data. Instead we apply the improved LMT-heuristic only to edges surviving an initial round of the normal LMT-heuristic.

\subsection{Parallelization}
\label{ssec::parallelization}

Because the LMT-heuristic performs only  local changes, most of the edges can be processed in parallel without synchronization. Problems occur only if adjacent edges are processed concurrently (for the improved LMT-skeleton this is unfortunately not true, because marking an edge \emph{impossible} affects a larger neighborhood of edges). In order to parallelize the normal LMT-heuristic, we implemented a solution based on data partitioning without any explicit locking.

We cut the vertices $V$ into two disjoint sets $V= V_1 \cup V_2$ and process only those edges with both endpoints in $V_1$ (resp.~$V_2$) in parallel. Define $C$ as the cut set ${\{\{s,t\}\in E\mid s\in V_1,t\in V_2\}}$, i.e., all edges with one endpoint in $V_1$ and the other in $V_2$. 
While edges in $E(V_1)$ resp.~$E(V_2)$ are processed in parallel by two threads, edges in $C$ are accessed read-only by both threads and are handled after both threads join. This way we never process two edges with a common endpoint in parallel.

This leaves the question of how to partition the vertices into two disjoint sets. Recall that all vertices are stored in contiguous memory and are sorted in Hilbert order. A split in the middle of the array partitions the points into two sets that are separated by a rather simple curve. Therefore, the cut set is likely to be small. Our half-edge array is sorted by source vertex, i.e., getting all edges with a specific source vertex in either half of the partition is trivial. Deciding if an edge $e=(s,t)$ is in the cut set consists of two comparisons of pointer $t$ against the lower and upper bound of the vertex subset.
Furthermore, with the fair assumption that the average degree of vertices is the same in both partitions, we obtain perfectly balanced partitions w.r.t.~the number of edges.

In order to avoid a serial scan at the top, we push the actual work of computing $C$ down to the leaves in the recursion tree. Scanning of the half-edge array starts at the leave nodes: processing of half-edges that belong to some cut set is postponed, instead they are passed back to the parent node. The parent in turn scans the edges it got from its two children, processes all edges it can and passes up the remaining ones. In other words, the final cut set $C$ bubbles up in the tree, while all intermediate cuts are never explicitly computed. The edges passed up from a node typically contain half-edges of several higher-level cuts.
This way, partitioning on each level of the recursion tree only takes constant time, while the actual work is fully parallelized at the leaf level.

Experiments and observations indicate that on large, uniformly distributed point sets approximately 0.15\% of all edges make it back to the root node, i.e., the amount of serial processing is low and the approach scales well. 
On degenerate instances it can perform poorly; e.g. if all points lie on a circle, then half of the edges will be returned to the root. For such cases, the code could be extended to repartition the remaining edges with another set of cuts.

After the LMT-heuristic completes, we are left with many polygonal faces that still need to be triangulated. Our implementation traverses the graph formed by the edges with one producer thread in order to collect all faces and multiple consumer threads to triangulate them with dynamic programming.

\section{Computational Results}
\label{sec::experiments}

Computations were performed on a machine with an Intel i7-6700K quad-core and 64GB memory. The code was written in C++ and compiled with gcc 5.4.0. 

We utilized CGAL \cite{cgal} for its exact orientation predicates, however, parts of the code are still prone to numerical errors. For example, triangulating the remaining polygonal faces requires to compute and compare the sum of radicals, which we implemented with double-precision arithmetic.
For small instances, it was possible to compare the results of our implementation against an independent implementation based on an integer programming formulation of the MWT problem. However, straightforward integer programming becomes infeasible quite fast and comparisons for point sets with thousands of points were not possible. 

\subsection{Uniformly and Normally Distributed Point Sets}

\begin{table}
\includegraphics[width=\columnwidth]{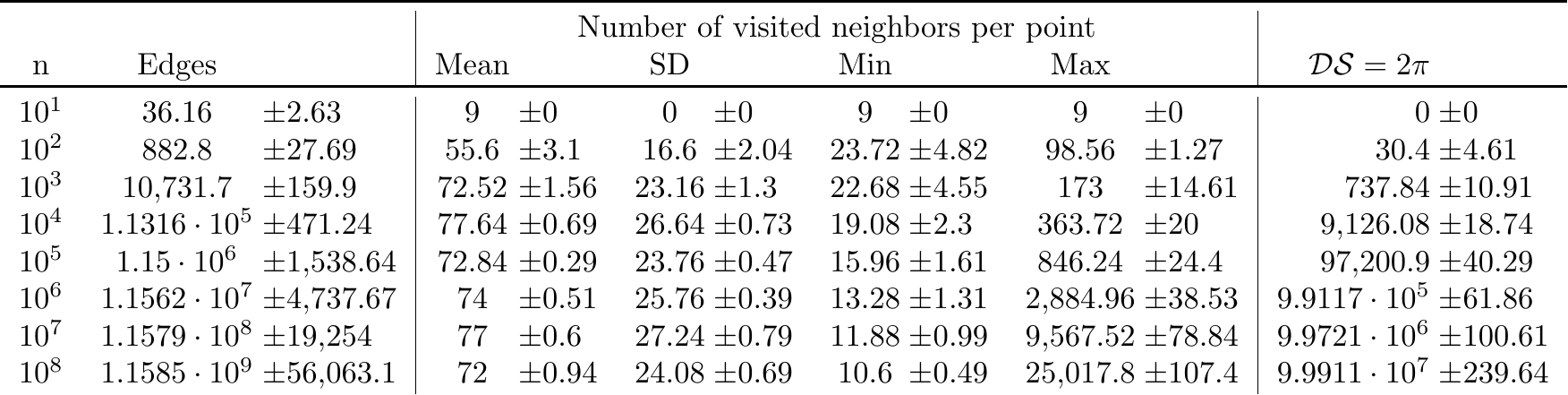}
\caption[Diamond test statistics]{Diamond test implementation on uniformly distributed point sets. The table shows the mean and the standard deviation of 25 different instances. The extreme values are assumed by points at the point set boundary.}
\label{tbl::diamond_uniform_1}
\end{table}

Table \ref{tbl::diamond_uniform_1} shows results of our diamond test implementation on uniformly distributed point sets with sizes ranging from $10$ to $10^8$ points. The table shows the mean values and the standard deviation of 25 different instances. Each instance was generated by choosing $n$ points uniformly from a square centered at the origin. Point coordinates were double-precision values. The diamond test performs one incremental nearest neighbor query for each point in order to generate the edges that pass the test. The last column shows the number of queries that were aborted early because dead sector covered the whole search space.

\begin{table}
\includegraphics[width=\columnwidth]{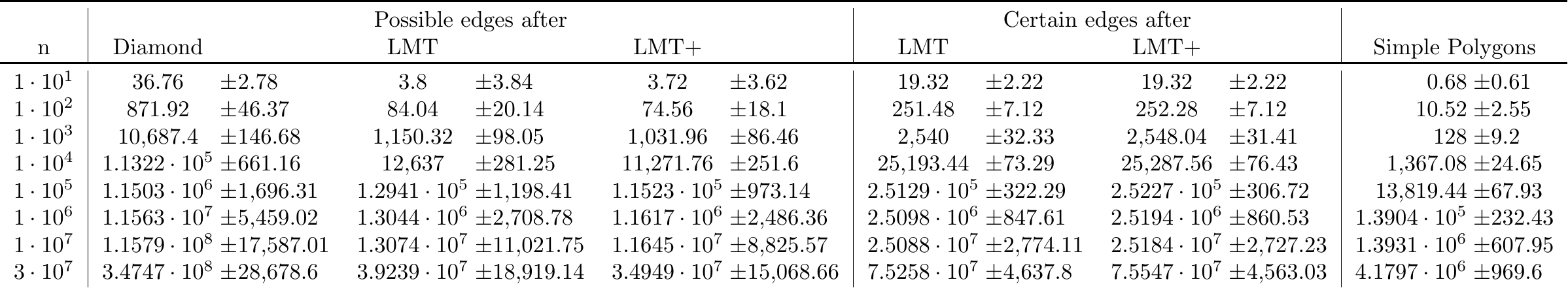}
\caption{LMT-skeleton statistics on uniformly distributed point sets.}
\label{tbl::lmt_uniform}
\end{table}

Table \ref{tbl::lmt_uniform} shows statistics for the LMT-heuristic on uniformly distributed point sets. The instance sizes range from 10 points up to 30,000,000 points. For each size 25 different instances were generated.
For the largest instances, the array storing the half-edges consumes nearly 39 GB of memory on its own. The serial initialization of the half-edge data structure, which basically amounts to radially sorting edges, takes longer than the parallel \LmtLoop on uniformly and normally distributed points.
The improved LMT-skeleton by Aichholzer et al. is denoted LMT+ in the tables. The resulting skeleton was almost always connected in the computations and the number of remaining simple polygons that needed to be triangulated is shown in the last column. Only one instance of size $3 \cdot 10^7$ contained one non-simple polygon.

As we can see, the LMT-skeleton eliminates most of the possible edges with only $\approx 11\% $ remaining. Given that any triangulation has $3n - |\CH| - 3$ edges, the certain edges amount to $\approx 83\%$ of the complete triangulation. The improved LMT-skeleton reduces the amount of possible edges by another 10\%, but it provides hardly any additional certain edges.

\begin{figure}
\includegraphics[width=\columnwidth]{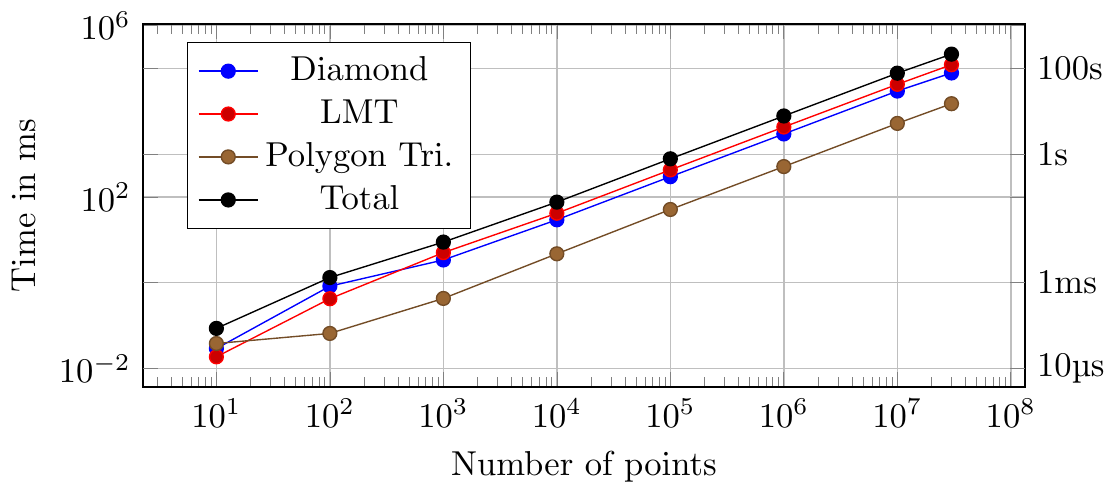}
\caption{LMT-skeleton runtime on uniformly distributed point sets.}
\label{fig::lmt_runtime}
\end{figure}

The results on normally distributed point sets are basically identical. Point coordinates were generated by two normally distributed random variables $ X,Y \sim \mathcal{N}(\mu,\,\sigma^{2})$, with mean $\mu = 0$ and standard deviation $\sigma \in \{1,100,100000\}$. The tables are given in Appendix \ref{app::gaussian}.

\subsection{TSPLIB + VLSI}
\label{ssec::tsplib}

In addition to uniformly and normally distributed instances, we ran our implementation on instances found in the well-known TSPLIB \cite{reinelt1991}, which contains a wide variety of instances with different distributions. The instances are drawn from industrial applications and from geographic problems.
All 94 instances have a connected LMT-skeleton and can be fully triangulated with dynamic programming to obtain the minimum weight triangulation. The total time it took to solve all instances of the TSPLIB was approximately 8.5 seconds. A complete breakdown for each instance is given in Appendix \ref{app::tsplib} Table~\ref{tbl::tsplib}.

Additional point sets can be found at \texttt{\url{http://www.math.uwaterloo.ca/tsp/vlsi/}}. This collection of 102 TSP instances was provided by Andre Rohe, based on VLSI data sets studied at the Forschungsinstitut für Diskrete Mathematik, Universität Bonn. 
The LMT-heuristic is sufficient to solve all instances, except \texttt{lra498378}, which contained two non-simple polygonal faces.
A complete breakdown is given in Appendix \ref{app::vlsi}. 
Our implementation of the improved LMT-skeleton performs exceedingly bad on some of these instances; see Table~\ref{tbl::vlsi}.  
These instances contain empty regions with many points on the ``boundary''. Such regions are the worst-case for the heuristics because most edges inside them have the diamond property, which in turn leads to vertices with very high degree. Whenever an edge is found to be impossible by the improved LMT-skeleton, almost all edges are restacked and rechecked. 
Given the overall results of the improved LMT-skeleton, storing additional data to increase performance and/or limiting it to non-simple polygons may be reasonable.

\bgroup \small \setlength {\tabcolsep }{4pt}\begin {longtabu}{c|r|r|r|r|r|r}%
\caption {VLSI statistics}\label{tbl::vlsi}\\ & \multicolumn {6}{c}{Time in ms}\\Instance&Total&DT&LMT-Init&LMT-Loop&LMT+&Dyn. Prog.\\
\hline
ara238025&\pgfutilensuremath {15{,}325}&\pgfutilensuremath {4{,}954}&\pgfutilensuremath {446}&\pgfutilensuremath {496}&\pgfutilensuremath {9{,}279}&\pgfutilensuremath {148}\\%
lra498378&\pgfutilensuremath {382{,}932}&\pgfutilensuremath {44{,}267}&\pgfutilensuremath {1{,}238}&\pgfutilensuremath {7{,}532}&\pgfutilensuremath {329{,}292}&\pgfutilensuremath {599}\\%
lrb744710&\pgfutilensuremath {484{,}430}&\pgfutilensuremath {7{,}952}&\pgfutilensuremath {1{,}377}&\pgfutilensuremath {2{,}661}&\pgfutilensuremath {471{,}564}&\pgfutilensuremath {872}\\%
sra104815&\pgfutilensuremath {1{,}937}&\pgfutilensuremath {559}&\pgfutilensuremath {191}&\pgfutilensuremath {198}&\pgfutilensuremath {922}&\pgfutilensuremath {65}\\%
\end {longtabu}\egroup %

\section{Conclusion}
We have shown that despite of the theoretical hardness of the MWT problem, a wide range of large-scale instances can be solved to optimality.

Difficulties for other instances arise from two sources. On one hand, we have instances containing  more or less regular $k$-gons with one or more points near the center. These configurations can lead to a highly disconnected LMT-skeleton (an example is given by Belleville et al.~\cite{belleville96}) and require exponential time algorithms to complete the MWT. Preliminary experiments suggest that such configurations are best solved with integer programming. The example point set given by Belleville et al.~\cite{belleville96} can easily be solved with CPLEX in less than a minute, while the dynamic programming implementation of Grantson et al.~\cite{grantsonBL08} was not able to solve it within several hours.
On the other hand, we have instances containing empty regions with many points on their ``boundary'', such as empty $k$-gons and circles. They may be solvable in polynomial time, but trigger the worst-case behavior of the heuristics.  
Deciding what is the best approach to handle these two types of difficulties and integrating it into our implementation is left for future work.  

\subparagraph*{Acknowledgements.}

I want to thank Sándor Fekete and Victor Alvarez for useful discussions and suggestions that helped to improve the presentation of this paper.

\newpage
\bibliography{bibliography}
\newpage
\appendix
\section{Normally distributed point sets}
\label{app::gaussian}

Point coordinates were independently generated by two normal distributed random variables $ X,Y \sim \mathcal{N}(\mu,\,\sigma^{2})$, with mean $\mu = 0$ and standard deviation $\sigma \in \{1,100,100000\}$.

\begin{table}[!h]
	\includegraphics[width=\columnwidth]{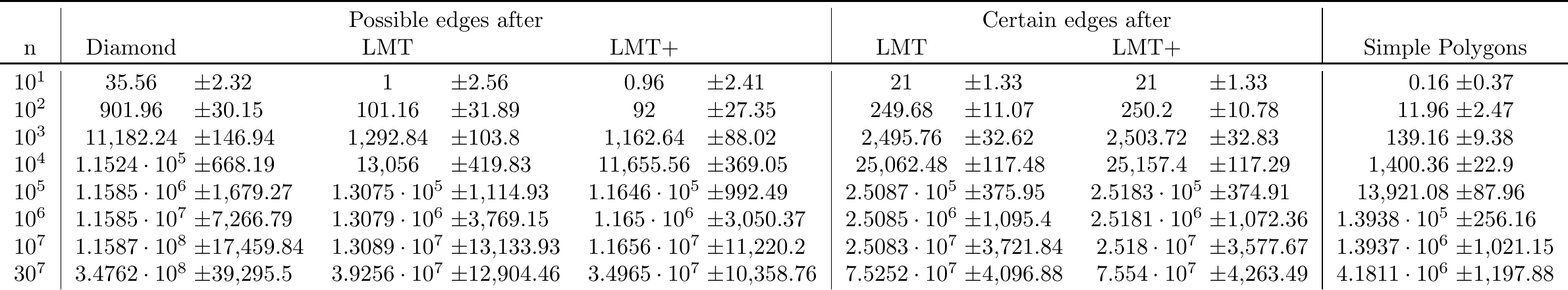}
	\caption{$X,Y \sim \mathcal{N}(0,\,1)$}
\end{table}

\begin{table}[!h]
	\includegraphics[width=\columnwidth]{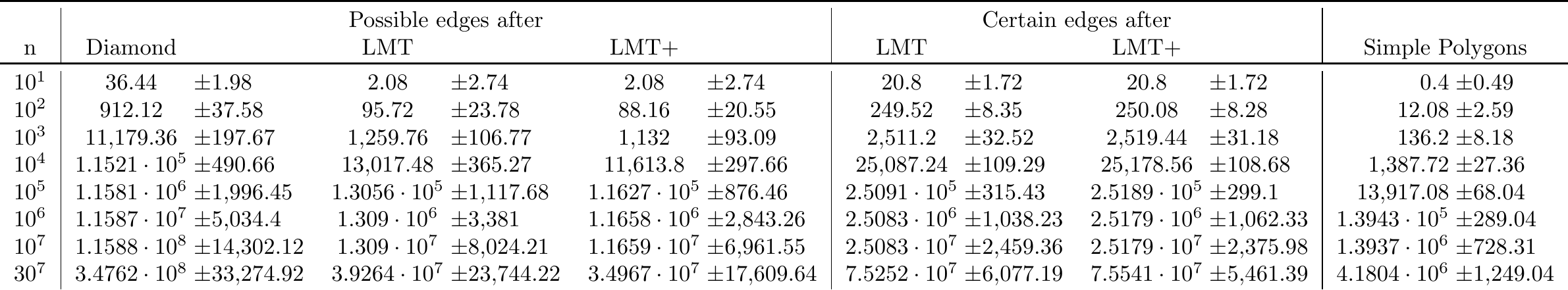}
	\caption{$X,Y \sim \mathcal{N}(0,\,1000)$}
\end{table}

\begin{table}[!h]
	\includegraphics[width=\columnwidth]{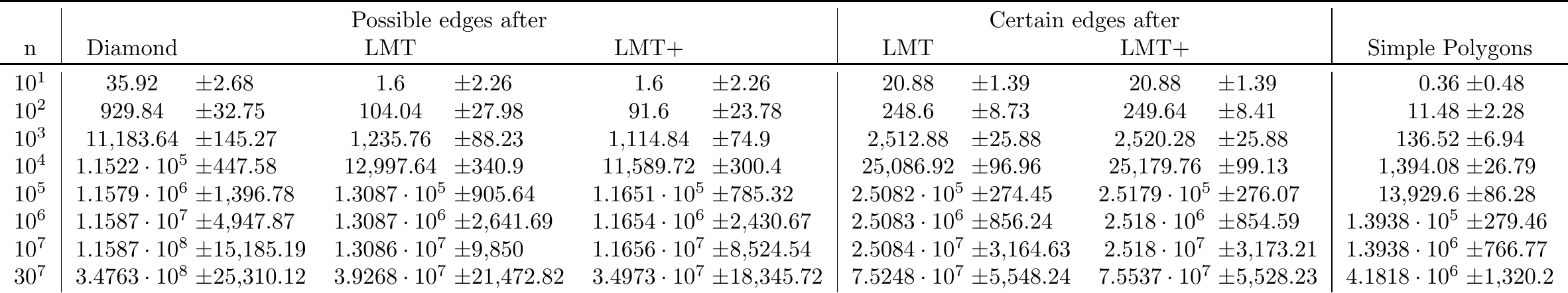}
	\caption{$X,Y \sim \mathcal{N}(0,\,1000000)$}
\end{table}

\newpage
\section{TSPLIB Results}
\label{app::tsplib}

The following table shows the results on instances of the TSPLIB. All instances have a connected LMT-skeleton and can be completed to the MWT with dynamic programming. The weight ratio of the Delaunay triangulation and the  MWT is given in the second column.
The time includes the time that was needed for the subsequent polygon triangulation with dynamic programming. 
Examples of two instances are given in Appendix \ref{app::tsplib_examples}.

\begingroup \small %
\bgroup \small \setlength {\tabcolsep }{4pt}\begin {longtabu}{c|r|rrRR|RR|r@{}l}%
\caption {TSPLIB statistics}\label {tbl::tsplib}\\ \toprule & \multicolumn {1}{c|}{Weight} &\multicolumn {4}{c|}{Possible edges after} & \multicolumn {2}{c|}{Certain edges}\\Instance&Ratio&DT&Factor&LMT&LMT+&LMT&LMT+&\multicolumn {2}{c}{Time in ms}\\\toprule \endhead \midrule \multicolumn {10}{r}{{Continued on next page}} \\ \bottomrule \endfoot \midrule \multicolumn {10}{r}{{End}} \\ \bottomrule \endlastfoot %
a280&\pgfutilensuremath {1.014}&\pgfutilensuremath {2{,}444}&\pgfutilensuremath {8.76}&\pgfutilensuremath {414}&\pgfutilensuremath {379}&\pgfutilensuremath {642}&\pgfutilensuremath {643}&\pgfmathprintnumber [fixed,fixed zerofill,precision=2]{1Y2.61533e0]}& $\mskip \thickmuskip \pm $\pgfmathprintnumber [fixed,fixed zerofill,precision=2]{1Y1.58586e-1]}\\%
ali535&\pgfutilensuremath {1.056}&\pgfutilensuremath {5{,}936}&\pgfutilensuremath {11.731}&\pgfutilensuremath {806}&\pgfutilensuremath {702}&\pgfutilensuremath {1{,}227}&\pgfutilensuremath {1{,}236}&\pgfmathprintnumber [fixed,fixed zerofill,precision=2]{1Y5.94965e0]}& $\mskip \thickmuskip \pm $\pgfmathprintnumber [fixed,fixed zerofill,precision=2]{1Y2.73186e-1]}\\%
att48&\pgfutilensuremath {1.014}&\pgfutilensuremath {305}&\pgfutilensuremath {6.354}&\pgfutilensuremath {11}&\pgfutilensuremath {11}&\pgfutilensuremath {125}&\pgfutilensuremath {125}&\pgfmathprintnumber [fixed,fixed zerofill,precision=2]{1Y4.23277e-1]}& $\mskip \thickmuskip \pm $\pgfmathprintnumber [fixed,fixed zerofill,precision=2]{1Y2.93195e-2]}\\%
att532&\pgfutilensuremath {1.039}&\pgfutilensuremath {5{,}507}&\pgfutilensuremath {10.352}&\pgfutilensuremath {561}&\pgfutilensuremath {491}&\pgfutilensuremath {1{,}370}&\pgfutilensuremath {1{,}375}&\pgfmathprintnumber [fixed,fixed zerofill,precision=2]{1Y4.28769e0]}& $\mskip \thickmuskip \pm $\pgfmathprintnumber [fixed,fixed zerofill,precision=2]{1Y8.19713e-2]}\\%
berlin52&\pgfutilensuremath {1.022}&\pgfutilensuremath {416}&\pgfutilensuremath {8}&\pgfutilensuremath {71}&\pgfutilensuremath {67}&\pgfutilensuremath {116}&\pgfutilensuremath {116}&\pgfmathprintnumber [fixed,fixed zerofill,precision=2]{1Y6.01144e-1]}& $\mskip \thickmuskip \pm $\pgfmathprintnumber [fixed,fixed zerofill,precision=2]{1Y4.2973e-2]}\\%
bier127&\pgfutilensuremath {1.024}&\pgfutilensuremath {1{,}207}&\pgfutilensuremath {9.504}&\pgfutilensuremath {86}&\pgfutilensuremath {82}&\pgfutilensuremath {330}&\pgfutilensuremath {330}&\pgfmathprintnumber [fixed,fixed zerofill,precision=2]{1Y1.494e0]}& $\mskip \thickmuskip \pm $\pgfmathprintnumber [fixed,fixed zerofill,precision=2]{1Y1.11136e-2]}\\%
brd14051&\pgfutilensuremath {1.026}&\pgfutilensuremath {143{,}816}&\pgfutilensuremath {10.235}&\pgfutilensuremath {11{,}800}&\pgfutilensuremath {10{,}573}&\pgfutilensuremath {37{,}562}&\pgfutilensuremath {37{,}633}&\pgfmathprintnumber [fixed,fixed zerofill,precision=2]{1Y1.00435e2]}& $\mskip \thickmuskip \pm $\pgfmathprintnumber [fixed,fixed zerofill,precision=2]{1Y2.75723e0]}\\%
burma14&\pgfutilensuremath {1.002}&\pgfutilensuremath {66}&\pgfutilensuremath {4.714}&\pgfutilensuremath {0}&\pgfutilensuremath {0}&\pgfutilensuremath {34}&\pgfutilensuremath {34}&\pgfmathprintnumber [fixed,fixed zerofill,precision=2]{1Y4.00409e-2]}& $\mskip \thickmuskip \pm $\pgfmathprintnumber [fixed,fixed zerofill,precision=2]{1Y1.66702e-2]}\\%
ch130&\pgfutilensuremath {1.039}&\pgfutilensuremath {1{,}246}&\pgfutilensuremath {9.585}&\pgfutilensuremath {132}&\pgfutilensuremath {117}&\pgfutilensuremath {334}&\pgfutilensuremath {334}&\pgfmathprintnumber [fixed,fixed zerofill,precision=2]{1Y1.49934e0]}& $\mskip \thickmuskip \pm $\pgfmathprintnumber [fixed,fixed zerofill,precision=2]{1Y3.90784e-2]}\\%
ch150&\pgfutilensuremath {1.02}&\pgfutilensuremath {1{,}340}&\pgfutilensuremath {8.933}&\pgfutilensuremath {151}&\pgfutilensuremath {140}&\pgfutilensuremath {367}&\pgfutilensuremath {367}&\pgfmathprintnumber [fixed,fixed zerofill,precision=2]{1Y1.47281e0]}& $\mskip \thickmuskip \pm $\pgfmathprintnumber [fixed,fixed zerofill,precision=2]{1Y4.84561e-2]}\\%
d198&\pgfutilensuremath {1.078}&\pgfutilensuremath {1{,}848}&\pgfutilensuremath {9.333}&\pgfutilensuremath {203}&\pgfutilensuremath {184}&\pgfutilensuremath {484}&\pgfutilensuremath {486}&\pgfmathprintnumber [fixed,fixed zerofill,precision=2]{1Y2.44951e0]}& $\mskip \thickmuskip \pm $\pgfmathprintnumber [fixed,fixed zerofill,precision=2]{1Y1.04141e-1]}\\%
d493&\pgfutilensuremath {1.049}&\pgfutilensuremath {4{,}831}&\pgfutilensuremath {9.799}&\pgfutilensuremath {351}&\pgfutilensuremath {319}&\pgfutilensuremath {1{,}322}&\pgfutilensuremath {1{,}326}&\pgfmathprintnumber [fixed,fixed zerofill,precision=2]{1Y3.58618e0]}& $\mskip \thickmuskip \pm $\pgfmathprintnumber [fixed,fixed zerofill,precision=2]{1Y1.32486e-1]}\\%
d657&\pgfutilensuremath {1.04}&\pgfutilensuremath {7{,}209}&\pgfutilensuremath {10.973}&\pgfutilensuremath {889}&\pgfutilensuremath {784}&\pgfutilensuremath {1{,}644}&\pgfutilensuremath {1{,}649}&\pgfmathprintnumber [fixed,fixed zerofill,precision=2]{1Y6.6063e0]}& $\mskip \thickmuskip \pm $\pgfmathprintnumber [fixed,fixed zerofill,precision=2]{1Y1.75412e0]}\\%
d1291&\pgfutilensuremath {1.015}&\pgfutilensuremath {15{,}704}&\pgfutilensuremath {12.164}&\pgfutilensuremath {2{,}437}&\pgfutilensuremath {2{,}318}&\pgfutilensuremath {2{,}821}&\pgfutilensuremath {2{,}825}&\pgfmathprintnumber [fixed,fixed zerofill,precision=2]{1Y1.47595e1]}& $\mskip \thickmuskip \pm $\pgfmathprintnumber [fixed,fixed zerofill,precision=2]{1Y5.81155e-1]}\\%
d1655&\pgfutilensuremath {1.061}&\pgfutilensuremath {18{,}238}&\pgfutilensuremath {11.02}&\pgfutilensuremath {2{,}486}&\pgfutilensuremath {2{,}331}&\pgfutilensuremath {3{,}839}&\pgfutilensuremath {3{,}857}&\pgfmathprintnumber [fixed,fixed zerofill,precision=2]{1Y1.6423e1]}& $\mskip \thickmuskip \pm $\pgfmathprintnumber [fixed,fixed zerofill,precision=2]{1Y7.20009e-1]}\\%
d2103&\pgfutilensuremath {1.138}&\pgfutilensuremath {18{,}575}&\pgfutilensuremath {8.833}&\pgfutilensuremath {3{,}655}&\pgfutilensuremath {3{,}626}&\pgfutilensuremath {4{,}504}&\pgfutilensuremath {4{,}504}&\pgfmathprintnumber [fixed,fixed zerofill,precision=2]{1Y1.36989e1]}& $\mskip \thickmuskip \pm $\pgfmathprintnumber [fixed,fixed zerofill,precision=2]{1Y2.4949e-1]}\\%
d15112&\pgfutilensuremath {1.022}&\pgfutilensuremath {153{,}329}&\pgfutilensuremath {10.146}&\pgfutilensuremath {12{,}354}&\pgfutilensuremath {11{,}020}&\pgfutilensuremath {40{,}513}&\pgfutilensuremath {40{,}611}&\pgfmathprintnumber [fixed,fixed zerofill,precision=2]{1Y1.02385e2]}& $\mskip \thickmuskip \pm $\pgfmathprintnumber [fixed,fixed zerofill,precision=2]{1Y2.53508e0]}\\%
d18512&\pgfutilensuremath {1.021}&\pgfutilensuremath {183{,}951}&\pgfutilensuremath {9.937}&\pgfutilensuremath {14{,}559}&\pgfutilensuremath {13{,}135}&\pgfutilensuremath {49{,}729}&\pgfutilensuremath {49{,}819}&\pgfmathprintnumber [fixed,fixed zerofill,precision=2]{1Y1.13316e2]}& $\mskip \thickmuskip \pm $\pgfmathprintnumber [fixed,fixed zerofill,precision=2]{1Y1.10664e0]}\\%
dsj1000&\pgfutilensuremath {1.044}&\pgfutilensuremath {12{,}140}&\pgfutilensuremath {12.14}&\pgfutilensuremath {1{,}239}&\pgfutilensuremath {1{,}094}&\pgfutilensuremath {2{,}504}&\pgfutilensuremath {2{,}518}&\pgfmathprintnumber [fixed,fixed zerofill,precision=2]{1Y1.16682e1]}& $\mskip \thickmuskip \pm $\pgfmathprintnumber [fixed,fixed zerofill,precision=2]{1Y6.36674e-1]}\\%
eil51&\pgfutilensuremath {1.004}&\pgfutilensuremath {320}&\pgfutilensuremath {6.275}&\pgfutilensuremath {2}&\pgfutilensuremath {2}&\pgfutilensuremath {139}&\pgfutilensuremath {139}&\pgfmathprintnumber [fixed,fixed zerofill,precision=2]{1Y4.25812e-1]}& $\mskip \thickmuskip \pm $\pgfmathprintnumber [fixed,fixed zerofill,precision=2]{1Y1.47313e-2]}\\%
eil76&\pgfutilensuremath {1.014}&\pgfutilensuremath {509}&\pgfutilensuremath {6.697}&\pgfutilensuremath {20}&\pgfutilensuremath {20}&\pgfutilensuremath {205}&\pgfutilensuremath {205}&\pgfmathprintnumber [fixed,fixed zerofill,precision=2]{1Y7.82104e-1]}& $\mskip \thickmuskip \pm $\pgfmathprintnumber [fixed,fixed zerofill,precision=2]{1Y3.72772e-2]}\\%
eil101&\pgfutilensuremath {1.01}&\pgfutilensuremath {720}&\pgfutilensuremath {7.129}&\pgfutilensuremath {47}&\pgfutilensuremath {40}&\pgfutilensuremath {269}&\pgfutilensuremath {270}&\pgfmathprintnumber [fixed,fixed zerofill,precision=2]{1Y7.5505e-1]}& $\mskip \thickmuskip \pm $\pgfmathprintnumber [fixed,fixed zerofill,precision=2]{1Y5.76622e-2]}\\%
fl417&\pgfutilensuremath {1.123}&\pgfutilensuremath {5{,}313}&\pgfutilensuremath {12.741}&\pgfutilensuremath {331}&\pgfutilensuremath {297}&\pgfutilensuremath {1{,}065}&\pgfutilensuremath {1{,}066}&\pgfmathprintnumber [fixed,fixed zerofill,precision=2]{1Y5.07527e0]}& $\mskip \thickmuskip \pm $\pgfmathprintnumber [fixed,fixed zerofill,precision=2]{1Y1.57974e-1]}\\%
fl1400&\pgfutilensuremath {1.398}&\pgfutilensuremath {16{,}617}&\pgfutilensuremath {11.869}&\pgfutilensuremath {1{,}072}&\pgfutilensuremath {967}&\pgfutilensuremath {3{,}830}&\pgfutilensuremath {3{,}835}&\pgfmathprintnumber [fixed,fixed zerofill,precision=2]{1Y2.23723e1]}& $\mskip \thickmuskip \pm $\pgfmathprintnumber [fixed,fixed zerofill,precision=2]{1Y9.04692e-1]}\\%
fl1577&\pgfutilensuremath {1.137}&\pgfutilensuremath {45{,}012}&\pgfutilensuremath {28.543}&\pgfutilensuremath {2{,}753}&\pgfutilensuremath {2{,}500}&\pgfutilensuremath {3{,}998}&\pgfutilensuremath {4{,}012}&\pgfmathprintnumber [fixed,fixed zerofill,precision=2]{1Y6.20003e1]}& $\mskip \thickmuskip \pm $\pgfmathprintnumber [fixed,fixed zerofill,precision=2]{1Y3.10764e0]}\\%
fl3795&\pgfutilensuremath {1.385}&\pgfutilensuremath {156{,}648}&\pgfutilensuremath {41.277}&\pgfutilensuremath {19{,}655}&\pgfutilensuremath {15{,}297}&\pgfutilensuremath {8{,}368}&\pgfutilensuremath {8{,}404}&\pgfmathprintnumber [fixed,fixed zerofill,precision=2]{1Y1.231e3]}& $\mskip \thickmuskip \pm $\pgfmathprintnumber [fixed,fixed zerofill,precision=2]{1Y1.40369e1]}\\%
fnl4461&\pgfutilensuremath {1.019}&\pgfutilensuremath {42{,}765}&\pgfutilensuremath {9.586}&\pgfutilensuremath {3{,}114}&\pgfutilensuremath {2{,}847}&\pgfutilensuremath {12{,}081}&\pgfutilensuremath {12{,}103}&\pgfmathprintnumber [fixed,fixed zerofill,precision=2]{1Y2.29823e1]}& $\mskip \thickmuskip \pm $\pgfmathprintnumber [fixed,fixed zerofill,precision=2]{1Y2.98271e-1]}\\%
gil262&\pgfutilensuremath {1.035}&\pgfutilensuremath {2{,}650}&\pgfutilensuremath {10.115}&\pgfutilensuremath {220}&\pgfutilensuremath {212}&\pgfutilensuremath {681}&\pgfutilensuremath {681}&\pgfmathprintnumber [fixed,fixed zerofill,precision=2]{1Y2.29076e0]}& $\mskip \thickmuskip \pm $\pgfmathprintnumber [fixed,fixed zerofill,precision=2]{1Y1.64789e-1]}\\%
gr96&\pgfutilensuremath {1.039}&\pgfutilensuremath {963}&\pgfutilensuremath {10.031}&\pgfutilensuremath {140}&\pgfutilensuremath {121}&\pgfutilensuremath {224}&\pgfutilensuremath {225}&\pgfmathprintnumber [fixed,fixed zerofill,precision=2]{1Y1.504e0]}& $\mskip \thickmuskip \pm $\pgfmathprintnumber [fixed,fixed zerofill,precision=2]{1Y3.29832e-2]}\\%
gr137&\pgfutilensuremath {1.077}&\pgfutilensuremath {1{,}304}&\pgfutilensuremath {9.518}&\pgfutilensuremath {161}&\pgfutilensuremath {146}&\pgfutilensuremath {341}&\pgfutilensuremath {341}&\pgfmathprintnumber [fixed,fixed zerofill,precision=2]{1Y1.58575e0]}& $\mskip \thickmuskip \pm $\pgfmathprintnumber [fixed,fixed zerofill,precision=2]{1Y3.27884e-2]}\\%
gr202&\pgfutilensuremath {1.039}&\pgfutilensuremath {1{,}899}&\pgfutilensuremath {9.401}&\pgfutilensuremath {188}&\pgfutilensuremath {172}&\pgfutilensuremath {515}&\pgfutilensuremath {515}&\pgfmathprintnumber [fixed,fixed zerofill,precision=2]{1Y1.76384e0]}& $\mskip \thickmuskip \pm $\pgfmathprintnumber [fixed,fixed zerofill,precision=2]{1Y1.11066e-1]}\\%
gr229&\pgfutilensuremath {1.043}&\pgfutilensuremath {2{,}383}&\pgfutilensuremath {10.406}&\pgfutilensuremath {202}&\pgfutilensuremath {173}&\pgfutilensuremath {599}&\pgfutilensuremath {606}&\pgfmathprintnumber [fixed,fixed zerofill,precision=2]{1Y2.67052e0]}& $\mskip \thickmuskip \pm $\pgfmathprintnumber [fixed,fixed zerofill,precision=2]{1Y2.2706e-1]}\\%
gr431&\pgfutilensuremath {1.057}&\pgfutilensuremath {4{,}466}&\pgfutilensuremath {10.362}&\pgfutilensuremath {405}&\pgfutilensuremath {359}&\pgfutilensuremath {1{,}121}&\pgfutilensuremath {1{,}128}&\pgfmathprintnumber [fixed,fixed zerofill,precision=2]{1Y4.1327e0]}& $\mskip \thickmuskip \pm $\pgfmathprintnumber [fixed,fixed zerofill,precision=2]{1Y1.9644e-1]}\\%
gr666&\pgfutilensuremath {1.047}&\pgfutilensuremath {7{,}502}&\pgfutilensuremath {11.264}&\pgfutilensuremath {902}&\pgfutilensuremath {771}&\pgfutilensuremath {1{,}668}&\pgfutilensuremath {1{,}679}&\pgfmathprintnumber [fixed,fixed zerofill,precision=2]{1Y7.24063e0]}& $\mskip \thickmuskip \pm $\pgfmathprintnumber [fixed,fixed zerofill,precision=2]{1Y2.20699e-1]}\\%
kroA100&\pgfutilensuremath {1.029}&\pgfutilensuremath {923}&\pgfutilensuremath {9.23}&\pgfutilensuremath {51}&\pgfutilensuremath {49}&\pgfutilensuremath {263}&\pgfutilensuremath {263}&\pgfmathprintnumber [fixed,fixed zerofill,precision=2]{1Y1.28677e0]}& $\mskip \thickmuskip \pm $\pgfmathprintnumber [fixed,fixed zerofill,precision=2]{1Y8.37315e-3]}\\%
kroA150&\pgfutilensuremath {1.02}&\pgfutilensuremath {1{,}367}&\pgfutilensuremath {9.113}&\pgfutilensuremath {120}&\pgfutilensuremath {120}&\pgfutilensuremath {377}&\pgfutilensuremath {377}&\pgfmathprintnumber [fixed,fixed zerofill,precision=2]{1Y1.53256e0]}& $\mskip \thickmuskip \pm $\pgfmathprintnumber [fixed,fixed zerofill,precision=2]{1Y8.63528e-2]}\\%
kroA200&\pgfutilensuremath {1.023}&\pgfutilensuremath {1{,}934}&\pgfutilensuremath {9.67}&\pgfutilensuremath {146}&\pgfutilensuremath {130}&\pgfutilensuremath {528}&\pgfutilensuremath {529}&\pgfmathprintnumber [fixed,fixed zerofill,precision=2]{1Y2.08784e0]}& $\mskip \thickmuskip \pm $\pgfmathprintnumber [fixed,fixed zerofill,precision=2]{1Y2.0742e-2]}\\%
kroB100&\pgfutilensuremath {1.025}&\pgfutilensuremath {885}&\pgfutilensuremath {8.85}&\pgfutilensuremath {46}&\pgfutilensuremath {36}&\pgfutilensuremath {263}&\pgfutilensuremath {266}&\pgfmathprintnumber [fixed,fixed zerofill,precision=2]{1Y1.26265e0]}& $\mskip \thickmuskip \pm $\pgfmathprintnumber [fixed,fixed zerofill,precision=2]{1Y3.29427e-2]}\\%
kroB150&\pgfutilensuremath {1.025}&\pgfutilensuremath {1{,}402}&\pgfutilensuremath {9.347}&\pgfutilensuremath {149}&\pgfutilensuremath {134}&\pgfutilensuremath {382}&\pgfutilensuremath {382}&\pgfmathprintnumber [fixed,fixed zerofill,precision=2]{1Y1.55646e0]}& $\mskip \thickmuskip \pm $\pgfmathprintnumber [fixed,fixed zerofill,precision=2]{1Y4.41339e-2]}\\%
kroB200&\pgfutilensuremath {1.029}&\pgfutilensuremath {1{,}915}&\pgfutilensuremath {9{,}575}&\pgfutilensuremath {210}&\pgfutilensuremath {183}&\pgfutilensuremath {503}&\pgfutilensuremath {504}&\pgfmathprintnumber [fixed,fixed zerofill,precision=2]{1Y2.26466e0]}& $\mskip \thickmuskip \pm $\pgfmathprintnumber [fixed,fixed zerofill,precision=2]{1Y1.86949e-1]}\\%
kroC100&\pgfutilensuremath {1.014}&\pgfutilensuremath {870}&\pgfutilensuremath {8.7}&\pgfutilensuremath {100}&\pgfutilensuremath {86}&\pgfutilensuremath {252}&\pgfutilensuremath {252}&\pgfmathprintnumber [fixed,fixed zerofill,precision=2]{1Y1.26881e0]}& $\mskip \thickmuskip \pm $\pgfmathprintnumber [fixed,fixed zerofill,precision=2]{1Y1.24757e-2]}\\%
kroD100&\pgfutilensuremath {1.022}&\pgfutilensuremath {874}&\pgfutilensuremath {8.74}&\pgfutilensuremath {97}&\pgfutilensuremath {71}&\pgfutilensuremath {248}&\pgfutilensuremath {251}&\pgfmathprintnumber [fixed,fixed zerofill,precision=2]{1Y1.37196e0]}& $\mskip \thickmuskip \pm $\pgfmathprintnumber [fixed,fixed zerofill,precision=2]{1Y2.5419e-2]}\\%
kroE100&\pgfutilensuremath {1.026}&\pgfutilensuremath {896}&\pgfutilensuremath {8.96}&\pgfutilensuremath {123}&\pgfutilensuremath {113}&\pgfutilensuremath {242}&\pgfutilensuremath {242}&\pgfmathprintnumber [fixed,fixed zerofill,precision=2]{1Y1.41753e0]}& $\mskip \thickmuskip \pm $\pgfmathprintnumber [fixed,fixed zerofill,precision=2]{1Y2.59726e-2]}\\%
lin105&\pgfutilensuremath {1.022}&\pgfutilensuremath {854}&\pgfutilensuremath {8.133}&\pgfutilensuremath {159}&\pgfutilensuremath {148}&\pgfutilensuremath {237}&\pgfutilensuremath {238}&\pgfmathprintnumber [fixed,fixed zerofill,precision=2]{1Y1.03952e0]}& $\mskip \thickmuskip \pm $\pgfmathprintnumber [fixed,fixed zerofill,precision=2]{1Y7.48827e-2]}\\%
lin318&\pgfutilensuremath {1.021}&\pgfutilensuremath {3{,}300}&\pgfutilensuremath {10.377}&\pgfutilensuremath {584}&\pgfutilensuremath {541}&\pgfutilensuremath {732}&\pgfutilensuremath {735}&\pgfmathprintnumber [fixed,fixed zerofill,precision=2]{1Y3.14955e0]}& $\mskip \thickmuskip \pm $\pgfmathprintnumber [fixed,fixed zerofill,precision=2]{1Y2.23381e-1]}\\%
linhp318&\pgfutilensuremath {1.021}&\pgfutilensuremath {3{,}300}&\pgfutilensuremath {10.377}&\pgfutilensuremath {584}&\pgfutilensuremath {541}&\pgfutilensuremath {732}&\pgfutilensuremath {735}&\pgfmathprintnumber [fixed,fixed zerofill,precision=2]{1Y3.18427e0]}& $\mskip \thickmuskip \pm $\pgfmathprintnumber [fixed,fixed zerofill,precision=2]{1Y2.2147e-1]}\\%
nrw1379&\pgfutilensuremath {1.017}&\pgfutilensuremath {12{,}828}&\pgfutilensuremath {9.302}&\pgfutilensuremath {964}&\pgfutilensuremath {851}&\pgfutilensuremath {3{,}735}&\pgfutilensuremath {3{,}741}&\pgfmathprintnumber [fixed,fixed zerofill,precision=2]{1Y8.5696e0]}& $\mskip \thickmuskip \pm $\pgfmathprintnumber [fixed,fixed zerofill,precision=2]{1Y3.39676e-1]}\\%
p654&\pgfutilensuremath {1.061}&\pgfutilensuremath {7{,}039}&\pgfutilensuremath {10.763}&\pgfutilensuremath {794}&\pgfutilensuremath {788}&\pgfutilensuremath {1{,}441}&\pgfutilensuremath {1{,}441}&\pgfmathprintnumber [fixed,fixed zerofill,precision=2]{1Y8.23492e0]}& $\mskip \thickmuskip \pm $\pgfmathprintnumber [fixed,fixed zerofill,precision=2]{1Y6.25813e-1]}\\%
pcb442&\pgfutilensuremath {1.025}&\pgfutilensuremath {3{,}852}&\pgfutilensuremath {8.715}&\pgfutilensuremath {609}&\pgfutilensuremath {528}&\pgfutilensuremath {1{,}049}&\pgfutilensuremath {1{,}060}&\pgfmathprintnumber [fixed,fixed zerofill,precision=2]{1Y3.39983e0]}& $\mskip \thickmuskip \pm $\pgfmathprintnumber [fixed,fixed zerofill,precision=2]{1Y1.08191e-1]}\\%
pcb1173&\pgfutilensuremath {1.063}&\pgfutilensuremath {12{,}931}&\pgfutilensuremath {11.024}&\pgfutilensuremath {2{,}315}&\pgfutilensuremath {2{,}054}&\pgfutilensuremath {2{,}712}&\pgfutilensuremath {2{,}727}&\pgfmathprintnumber [fixed,fixed zerofill,precision=2]{1Y1.05156e1]}& $\mskip \thickmuskip \pm $\pgfmathprintnumber [fixed,fixed zerofill,precision=2]{1Y1.27158e-1]}\\%
pcb3038&\pgfutilensuremath {1.037}&\pgfutilensuremath {30{,}748}&\pgfutilensuremath {10.121}&\pgfutilensuremath {4{,}113}&\pgfutilensuremath {3{,}656}&\pgfutilensuremath {7{,}586}&\pgfutilensuremath {7{,}615}&\pgfmathprintnumber [fixed,fixed zerofill,precision=2]{1Y2.06836e1]}& $\mskip \thickmuskip \pm $\pgfmathprintnumber [fixed,fixed zerofill,precision=2]{1Y3.4466e-1]}\\%
pla7397&\pgfutilensuremath {1.031}&\pgfutilensuremath {94{,}964}&\pgfutilensuremath {12.838}&\pgfutilensuremath {37{,}070}&\pgfutilensuremath {34{,}494}&\pgfutilensuremath {15{,}684}&\pgfutilensuremath {15{,}731}&\pgfmathprintnumber [fixed,fixed zerofill,precision=2]{1Y6.04305e2]}& $\mskip \thickmuskip \pm $\pgfmathprintnumber [fixed,fixed zerofill,precision=2]{1Y1.60623e0]}\\%
pla33810&\pgfutilensuremath {1.052}&\pgfutilensuremath {403{,}528}&\pgfutilensuremath {11.935}&\pgfutilensuremath {79{,}699}&\pgfutilensuremath {73{,}726}&\pgfutilensuremath {74{,}897}&\pgfutilensuremath {75{,}105}&\pgfmathprintnumber [fixed,fixed zerofill,precision=2]{1Y6.34451e2]}& $\mskip \thickmuskip \pm $\pgfmathprintnumber [fixed,fixed zerofill,precision=2]{1Y2.67319e0]}\\%
pla85900&\pgfutilensuremath {1.045}&\pgfutilensuremath {983{,}759}&\pgfutilensuremath {11.452}&\pgfutilensuremath {292{,}864}&\pgfutilensuremath {264{,}826}&\pgfutilensuremath {188{,}411}&\pgfutilensuremath {189{,}095}&\pgfmathprintnumber [fixed,fixed zerofill,precision=2]{1Y4.1632e3]}& $\mskip \thickmuskip \pm $\pgfmathprintnumber [fixed,fixed zerofill,precision=2]{1Y9.95819e0]}\\%
pr76&\pgfutilensuremath {1.042}&\pgfutilensuremath {688}&\pgfutilensuremath {9.053}&\pgfutilensuremath {163}&\pgfutilensuremath {154}&\pgfutilensuremath {178}&\pgfutilensuremath {178}&\pgfmathprintnumber [fixed,fixed zerofill,precision=2]{1Y1.08707e0]}& $\mskip \thickmuskip \pm $\pgfmathprintnumber [fixed,fixed zerofill,precision=2]{1Y1.21832e-2]}\\%
pr107&\pgfutilensuremath {1.003}&\pgfutilensuremath {1{,}290}&\pgfutilensuremath {12.056}&\pgfutilensuremath {16}&\pgfutilensuremath {16}&\pgfutilensuremath {275}&\pgfutilensuremath {275}&\pgfmathprintnumber [fixed,fixed zerofill,precision=2]{1Y1.62462e0]}& $\mskip \thickmuskip \pm $\pgfmathprintnumber [fixed,fixed zerofill,precision=2]{1Y5.9719e-2]}\\%
pr124&\pgfutilensuremath {1.025}&\pgfutilensuremath {1{,}419}&\pgfutilensuremath {11.444}&\pgfutilensuremath {122}&\pgfutilensuremath {117}&\pgfutilensuremath {271}&\pgfutilensuremath {271}&\pgfmathprintnumber [fixed,fixed zerofill,precision=2]{1Y1.70563e0]}& $\mskip \thickmuskip \pm $\pgfmathprintnumber [fixed,fixed zerofill,precision=2]{1Y1.2761e-1]}\\%
pr136&\pgfutilensuremath {1.024}&\pgfutilensuremath {988}&\pgfutilensuremath {7.265}&\pgfutilensuremath {248}&\pgfutilensuremath {248}&\pgfutilensuremath {280}&\pgfutilensuremath {280}&\pgfmathprintnumber [fixed,fixed zerofill,precision=2]{1Y1.29662e0]}& $\mskip \thickmuskip \pm $\pgfmathprintnumber [fixed,fixed zerofill,precision=2]{1Y4.1591e-2]}\\%
pr144&\pgfutilensuremath {1.086}&\pgfutilensuremath {2{,}532}&\pgfutilensuremath {17.583}&\pgfutilensuremath {645}&\pgfutilensuremath {645}&\pgfutilensuremath {290}&\pgfutilensuremath {290}&\pgfmathprintnumber [fixed,fixed zerofill,precision=2]{1Y3.13065e0]}& $\mskip \thickmuskip \pm $\pgfmathprintnumber [fixed,fixed zerofill,precision=2]{1Y2.54256e-1]}\\%
pr152&\pgfutilensuremath {1.177}&\pgfutilensuremath {2{,}844}&\pgfutilensuremath {18.711}&\pgfutilensuremath {198}&\pgfutilensuremath {190}&\pgfutilensuremath {358}&\pgfutilensuremath {358}&\pgfmathprintnumber [fixed,fixed zerofill,precision=2]{1Y3.7204e0]}& $\mskip \thickmuskip \pm $\pgfmathprintnumber [fixed,fixed zerofill,precision=2]{1Y2.27602e-1]}\\%
pr226&\pgfutilensuremath {1.058}&\pgfutilensuremath {4{,}683}&\pgfutilensuremath {20.721}&\pgfutilensuremath {525}&\pgfutilensuremath {525}&\pgfutilensuremath {452}&\pgfutilensuremath {452}&\pgfmathprintnumber [fixed,fixed zerofill,precision=2]{1Y5.9464e0]}& $\mskip \thickmuskip \pm $\pgfmathprintnumber [fixed,fixed zerofill,precision=2]{1Y3.2348e-1]}\\%
pr264&\pgfutilensuremath {1.115}&\pgfutilensuremath {3{,}096}&\pgfutilensuremath {11.727}&\pgfutilensuremath {463}&\pgfutilensuremath {448}&\pgfutilensuremath {563}&\pgfutilensuremath {563}&\pgfmathprintnumber [fixed,fixed zerofill,precision=2]{1Y4.61812e0]}& $\mskip \thickmuskip \pm $\pgfmathprintnumber [fixed,fixed zerofill,precision=2]{1Y3.55467e-1]}\\%
pr299&\pgfutilensuremath {1.031}&\pgfutilensuremath {3{,}087}&\pgfutilensuremath {10.324}&\pgfutilensuremath {736}&\pgfutilensuremath {595}&\pgfutilensuremath {701}&\pgfutilensuremath {711}&\pgfmathprintnumber [fixed,fixed zerofill,precision=2]{1Y4.7602e0]}& $\mskip \thickmuskip \pm $\pgfmathprintnumber [fixed,fixed zerofill,precision=2]{1Y6.20938e-2]}\\%
pr439&\pgfutilensuremath {1.055}&\pgfutilensuremath {5{,}324}&\pgfutilensuremath {12.128}&\pgfutilensuremath {1{,}306}&\pgfutilensuremath {1{,}160}&\pgfutilensuremath {991}&\pgfutilensuremath {996}&\pgfmathprintnumber [fixed,fixed zerofill,precision=2]{1Y1.15715e1]}& $\mskip \thickmuskip \pm $\pgfmathprintnumber [fixed,fixed zerofill,precision=2]{1Y2.99205e-1]}\\%
pr1002&\pgfutilensuremath {1.031}&\pgfutilensuremath {11{,}106}&\pgfutilensuremath {11.084}&\pgfutilensuremath {1{,}507}&\pgfutilensuremath {1{,}294}&\pgfutilensuremath {2{,}467}&\pgfutilensuremath {2{,}482}&\pgfmathprintnumber [fixed,fixed zerofill,precision=2]{1Y9.15588e0]}& $\mskip \thickmuskip \pm $\pgfmathprintnumber [fixed,fixed zerofill,precision=2]{1Y2.13194e-1]}\\%
pr2392&\pgfutilensuremath {1.038}&\pgfutilensuremath {28{,}911}&\pgfutilensuremath {12.087}&\pgfutilensuremath {6{,}405}&\pgfutilensuremath {5{,}296}&\pgfutilensuremath {5{,}625}&\pgfutilensuremath {5{,}697}&\pgfmathprintnumber [fixed,fixed zerofill,precision=2]{1Y3.30673e1]}& $\mskip \thickmuskip \pm $\pgfmathprintnumber [fixed,fixed zerofill,precision=2]{1Y2.70193e-1]}\\%
rat99&\pgfutilensuremath {1.013}&\pgfutilensuremath {684}&\pgfutilensuremath {6.909}&\pgfutilensuremath {64}&\pgfutilensuremath {61}&\pgfutilensuremath {254}&\pgfutilensuremath {254}&\pgfmathprintnumber [fixed,fixed zerofill,precision=2]{1Y1.02316e0]}& $\mskip \thickmuskip \pm $\pgfmathprintnumber [fixed,fixed zerofill,precision=2]{1Y1.59508e-2]}\\%
rat195&\pgfutilensuremath {1.014}&\pgfutilensuremath {1{,}408}&\pgfutilensuremath {7.221}&\pgfutilensuremath {206}&\pgfutilensuremath {197}&\pgfutilensuremath {472}&\pgfutilensuremath {472}&\pgfmathprintnumber [fixed,fixed zerofill,precision=2]{1Y1.51244e0]}& $\mskip \thickmuskip \pm $\pgfmathprintnumber [fixed,fixed zerofill,precision=2]{1Y1.22603e-1]}\\%
rat575&\pgfutilensuremath {1.012}&\pgfutilensuremath {4{,}912}&\pgfutilensuremath {8.543}&\pgfutilensuremath {523}&\pgfutilensuremath {459}&\pgfutilensuremath {1{,}493}&\pgfutilensuremath {1{,}499}&\pgfmathprintnumber [fixed,fixed zerofill,precision=2]{1Y3.48265e0]}& $\mskip \thickmuskip \pm $\pgfmathprintnumber [fixed,fixed zerofill,precision=2]{1Y3.62114e-2]}\\%
rat783&\pgfutilensuremath {1.019}&\pgfutilensuremath {7{,}195}&\pgfutilensuremath {9.189}&\pgfutilensuremath {719}&\pgfutilensuremath {650}&\pgfutilensuremath {2{,}033}&\pgfutilensuremath {2{,}035}&\pgfmathprintnumber [fixed,fixed zerofill,precision=2]{1Y4.86667e0]}& $\mskip \thickmuskip \pm $\pgfmathprintnumber [fixed,fixed zerofill,precision=2]{1Y1.38988e-1]}\\%
rd100&\pgfutilensuremath {1.012}&\pgfutilensuremath {859}&\pgfutilensuremath {8.59}&\pgfutilensuremath {88}&\pgfutilensuremath {75}&\pgfutilensuremath {252}&\pgfutilensuremath {252}&\pgfmathprintnumber [fixed,fixed zerofill,precision=2]{1Y1.34099e0]}& $\mskip \thickmuskip \pm $\pgfmathprintnumber [fixed,fixed zerofill,precision=2]{1Y9.15567e-3]}\\%
rd400&\pgfutilensuremath {1.023}&\pgfutilensuremath {3{,}976}&\pgfutilensuremath {9.94}&\pgfutilensuremath {398}&\pgfutilensuremath {358}&\pgfutilensuremath {1{,}027}&\pgfutilensuremath {1{,}032}&\pgfmathprintnumber [fixed,fixed zerofill,precision=2]{1Y3.46382e0]}& $\mskip \thickmuskip \pm $\pgfmathprintnumber [fixed,fixed zerofill,precision=2]{1Y2.42596e-1]}\\%
rl1304&\pgfutilensuremath {1.053}&\pgfutilensuremath {22{,}761}&\pgfutilensuremath {17.455}&\pgfutilensuremath {2{,}570}&\pgfutilensuremath {2{,}066}&\pgfutilensuremath {3{,}134}&\pgfutilensuremath {3{,}162}&\pgfmathprintnumber [fixed,fixed zerofill,precision=2]{1Y2.28403e1]}& $\mskip \thickmuskip \pm $\pgfmathprintnumber [fixed,fixed zerofill,precision=2]{1Y1.13826e0]}\\%
rl1323&\pgfutilensuremath {1.048}&\pgfutilensuremath {20{,}755}&\pgfutilensuremath {15.688}&\pgfutilensuremath {2{,}720}&\pgfutilensuremath {2{,}209}&\pgfutilensuremath {3{,}175}&\pgfutilensuremath {3{,}202}&\pgfmathprintnumber [fixed,fixed zerofill,precision=2]{1Y2.12004e1]}& $\mskip \thickmuskip \pm $\pgfmathprintnumber [fixed,fixed zerofill,precision=2]{1Y1.00314e0]}\\%
rl1889&\pgfutilensuremath {1.078}&\pgfutilensuremath {35{,}382}&\pgfutilensuremath {18.731}&\pgfutilensuremath {6{,}792}&\pgfutilensuremath {5{,}604}&\pgfutilensuremath {4{,}231}&\pgfutilensuremath {4{,}276}&\pgfmathprintnumber [fixed,fixed zerofill,precision=2]{1Y6.41078e1]}& $\mskip \thickmuskip \pm $\pgfmathprintnumber [fixed,fixed zerofill,precision=2]{1Y3.6737e-1]}\\%
rl5915&\pgfutilensuremath {1.052}&\pgfutilensuremath {91{,}750}&\pgfutilensuremath {15.511}&\pgfutilensuremath {10{,}689}&\pgfutilensuremath {9{,}355}&\pgfutilensuremath {14{,}899}&\pgfutilensuremath {14{,}928}&\pgfmathprintnumber [fixed,fixed zerofill,precision=2]{1Y1.27925e2]}& $\mskip \thickmuskip \pm $\pgfmathprintnumber [fixed,fixed zerofill,precision=2]{1Y3.02507e0]}\\%
rl5934&\pgfutilensuremath {1.068}&\pgfutilensuremath {93{,}545}&\pgfutilensuremath {15.764}&\pgfutilensuremath {12{,}209}&\pgfutilensuremath {10{,}074}&\pgfutilensuremath {14{,}744}&\pgfutilensuremath {14{,}823}&\pgfmathprintnumber [fixed,fixed zerofill,precision=2]{1Y1.29424e2]}& $\mskip \thickmuskip \pm $\pgfmathprintnumber [fixed,fixed zerofill,precision=2]{1Y2.00204e0]}\\%
rl11849&\pgfutilensuremath {1.048}&\pgfutilensuremath {148{,}715}&\pgfutilensuremath {12.551}&\pgfutilensuremath {19{,}283}&\pgfutilensuremath {16{,}746}&\pgfutilensuremath {30{,}167}&\pgfutilensuremath {30{,}240}&\pgfmathprintnumber [fixed,fixed zerofill,precision=2]{1Y3.08931e2]}& $\mskip \thickmuskip \pm $\pgfmathprintnumber [fixed,fixed zerofill,precision=2]{1Y2.44415e0]}\\%
st70&\pgfutilensuremath {1.029}&\pgfutilensuremath {553}&\pgfutilensuremath {7.9}&\pgfutilensuremath {81}&\pgfutilensuremath {69}&\pgfutilensuremath {170}&\pgfutilensuremath {170}&\pgfmathprintnumber [fixed,fixed zerofill,precision=2]{1Y9.40079e-1]}& $\mskip \thickmuskip \pm $\pgfmathprintnumber [fixed,fixed zerofill,precision=2]{1Y3.59003e-2]}\\%
ts225&\pgfutilensuremath {1}&\pgfutilensuremath {3{,}408}&\pgfutilensuremath {15.147}&\pgfutilensuremath {2{,}592}&\pgfutilensuremath {2{,}144}&\pgfutilensuremath {240}&\pgfutilensuremath {304}&\pgfmathprintnumber [fixed,fixed zerofill,precision=2]{1Y3.05612e1]}& $\mskip \thickmuskip \pm $\pgfmathprintnumber [fixed,fixed zerofill,precision=2]{1Y2.24947e-1]}\\%
tsp225&\pgfutilensuremath {1.009}&\pgfutilensuremath {1{,}762}&\pgfutilensuremath {7.831}&\pgfutilensuremath {212}&\pgfutilensuremath {188}&\pgfutilensuremath {544}&\pgfutilensuremath {545}&\pgfmathprintnumber [fixed,fixed zerofill,precision=2]{1Y1.68788e0]}& $\mskip \thickmuskip \pm $\pgfmathprintnumber [fixed,fixed zerofill,precision=2]{1Y1.14519e-1]}\\%
u159&\pgfutilensuremath {1.019}&\pgfutilensuremath {1{,}447}&\pgfutilensuremath {9.101}&\pgfutilensuremath {245}&\pgfutilensuremath {219}&\pgfutilensuremath {351}&\pgfutilensuremath {351}&\pgfmathprintnumber [fixed,fixed zerofill,precision=2]{1Y1.98636e0]}& $\mskip \thickmuskip \pm $\pgfmathprintnumber [fixed,fixed zerofill,precision=2]{1Y1.68417e-1]}\\%
u574&\pgfutilensuremath {1.038}&\pgfutilensuremath {6{,}195}&\pgfutilensuremath {10.793}&\pgfutilensuremath {718}&\pgfutilensuremath {636}&\pgfutilensuremath {1{,}438}&\pgfutilensuremath {1{,}444}&\pgfmathprintnumber [fixed,fixed zerofill,precision=2]{1Y4.96243e0]}& $\mskip \thickmuskip \pm $\pgfmathprintnumber [fixed,fixed zerofill,precision=2]{1Y1.5826e-1]}\\%
u724&\pgfutilensuremath {1.018}&\pgfutilensuremath {6{,}762}&\pgfutilensuremath {9.34}&\pgfutilensuremath {691}&\pgfutilensuremath {628}&\pgfutilensuremath {1{,}844}&\pgfutilensuremath {1{,}847}&\pgfmathprintnumber [fixed,fixed zerofill,precision=2]{1Y5.01454e0]}& $\mskip \thickmuskip \pm $\pgfmathprintnumber [fixed,fixed zerofill,precision=2]{1Y1.12434e-1]}\\%
u1060&\pgfutilensuremath {1.049}&\pgfutilensuremath {11{,}885}&\pgfutilensuremath {11.212}&\pgfutilensuremath {1{,}304}&\pgfutilensuremath {1{,}173}&\pgfutilensuremath {2{,}683}&\pgfutilensuremath {2{,}691}&\pgfmathprintnumber [fixed,fixed zerofill,precision=2]{1Y9.32977e0]}& $\mskip \thickmuskip \pm $\pgfmathprintnumber [fixed,fixed zerofill,precision=2]{1Y4.456e-1]}\\%
u1432&\pgfutilensuremath {1.003}&\pgfutilensuremath {9{,}039}&\pgfutilensuremath {6.312}&\pgfutilensuremath {2{,}287}&\pgfutilensuremath {2{,}224}&\pgfutilensuremath {3{,}145}&\pgfutilensuremath {3{,}154}&\pgfmathprintnumber [fixed,fixed zerofill,precision=2]{1Y7.47617e0]}& $\mskip \thickmuskip \pm $\pgfmathprintnumber [fixed,fixed zerofill,precision=2]{1Y1.61044e-1]}\\%
u1817&\pgfutilensuremath {1.007}&\pgfutilensuremath {16{,}214}&\pgfutilensuremath {8.924}&\pgfutilensuremath {3{,}229}&\pgfutilensuremath {3{,}063}&\pgfutilensuremath {3{,}996}&\pgfutilensuremath {4{,}007}&\pgfmathprintnumber [fixed,fixed zerofill,precision=2]{1Y1.47311e1]}& $\mskip \thickmuskip \pm $\pgfmathprintnumber [fixed,fixed zerofill,precision=2]{1Y3.34294e-1]}\\%
u2152&\pgfutilensuremath {1.007}&\pgfutilensuremath {17{,}709}&\pgfutilensuremath {8.229}&\pgfutilensuremath {3{,}811}&\pgfutilensuremath {3{,}674}&\pgfutilensuremath {4{,}597}&\pgfutilensuremath {4{,}603}&\pgfmathprintnumber [fixed,fixed zerofill,precision=2]{1Y1.52531e1]}& $\mskip \thickmuskip \pm $\pgfmathprintnumber [fixed,fixed zerofill,precision=2]{1Y2.70145e-1]}\\%
u2319&\pgfutilensuremath {1}&\pgfutilensuremath {10{,}179}&\pgfutilensuremath {4.389}&\pgfutilensuremath {4{,}061}&\pgfutilensuremath {4{,}061}&\pgfutilensuremath {4{,}853}&\pgfutilensuremath {4{,}853}&\pgfmathprintnumber [fixed,fixed zerofill,precision=2]{1Y8.2631e0]}& $\mskip \thickmuskip \pm $\pgfmathprintnumber [fixed,fixed zerofill,precision=2]{1Y1.0585e-1]}\\%
ulysses16&\pgfutilensuremath {1.014}&\pgfutilensuremath {74}&\pgfutilensuremath {4{,}625}&\pgfutilensuremath {0}&\pgfutilensuremath {0}&\pgfutilensuremath {38}&\pgfutilensuremath {38}&\pgfmathprintnumber [fixed,fixed zerofill,precision=2]{1Y6.49023e-2]}& $\mskip \thickmuskip \pm $\pgfmathprintnumber [fixed,fixed zerofill,precision=2]{1Y1.96262e-2]}\\%
ulysses22&\pgfutilensuremath {1.027}&\pgfutilensuremath {132}&\pgfutilensuremath {6}&\pgfutilensuremath {0}&\pgfutilensuremath {0}&\pgfutilensuremath {56}&\pgfutilensuremath {56}&\pgfmathprintnumber [fixed,fixed zerofill,precision=2]{1Y1.64044e-1]}& $\mskip \thickmuskip \pm $\pgfmathprintnumber [fixed,fixed zerofill,precision=2]{1Y1.93106e-2]}\\%
usa13509&\pgfutilensuremath {1.052}&\pgfutilensuremath {174{,}841}&\pgfutilensuremath {12.943}&\pgfutilensuremath {23{,}544}&\pgfutilensuremath {20{,}713}&\pgfutilensuremath {33{,}492}&\pgfutilensuremath {33{,}636}&\pgfmathprintnumber [fixed,fixed zerofill,precision=2]{1Y2.74644e2]}& $\mskip \thickmuskip \pm $\pgfmathprintnumber [fixed,fixed zerofill,precision=2]{1Y3.04782e0]}\\%
vm1084&\pgfutilensuremath {1.029}&\pgfutilensuremath {11{,}668}&\pgfutilensuremath {10.764}&\pgfutilensuremath {1{,}164}&\pgfutilensuremath {1{,}027}&\pgfutilensuremath {2{,}466}&\pgfutilensuremath {2{,}480}&\pgfmathprintnumber [fixed,fixed zerofill,precision=2]{1Y9.58794e0]}& $\mskip \thickmuskip \pm $\pgfmathprintnumber [fixed,fixed zerofill,precision=2]{1Y1.16919e-1]}\\%
vm1748&\pgfutilensuremath {1.034}&\pgfutilensuremath {20{,}712}&\pgfutilensuremath {11.849}&\pgfutilensuremath {2{,}848}&\pgfutilensuremath {2{,}307}&\pgfutilensuremath {4{,}032}&\pgfutilensuremath {4{,}049}&\pgfmathprintnumber [fixed,fixed zerofill,precision=2]{1Y2.53607e1]}& $\mskip \thickmuskip \pm $\pgfmathprintnumber [fixed,fixed zerofill,precision=2]{1Y8.55504e-1]}\\%
\end {longtabu}\egroup %
\endgroup %

\newpage
\section{TSPLIB Examples}
\label{app::tsplib_examples}

\begin{figure}[!h]
\includegraphics[width=0.88\columnwidth]{brd14051_lmt}
\caption{LMT-skeleton of \texttt{brd14501}.}
\label{fig::brd14501_lmt}
\end{figure}

\begin{figure}[!h]
\includegraphics[width=0.88\columnwidth]{brd14051_mwt}
\caption{MWT of \texttt{brd14501}.}
\end{figure}

\begin{figure}
\includegraphics[width=\columnwidth]{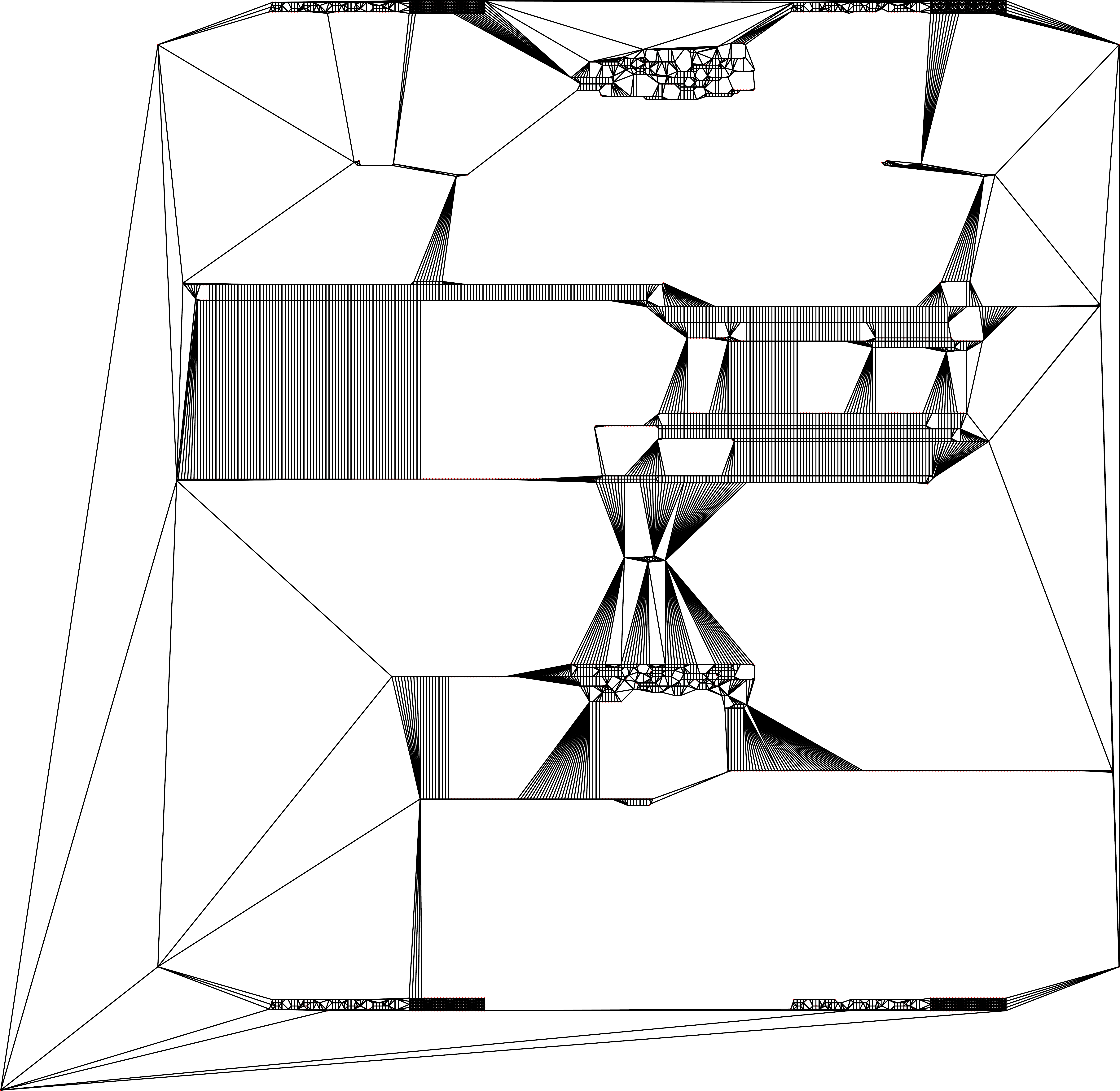}
\caption{LMT-skeleton of \texttt{fl3795}}
\end{figure}

\begin{figure}
\includegraphics[width=\columnwidth]{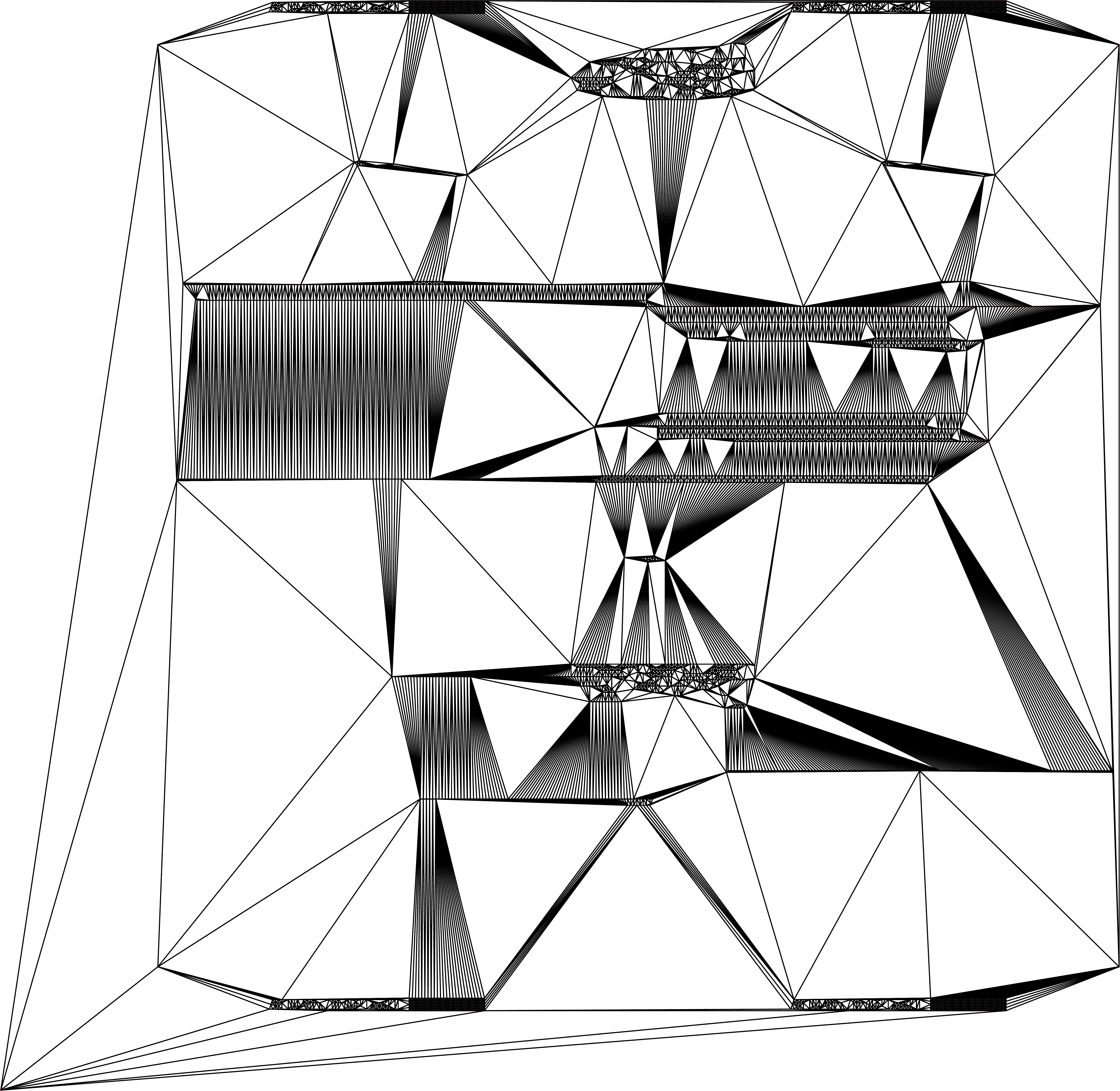}
\caption{MWT of \texttt{fl3795}}
\end{figure}

\newpage
\section{VLSI}
\label{app::vlsi}

The following table shows the results on VLSI point sets, which can be downloaded at \texttt{http://www.math.uwaterloo.ca/tsp/vlsi/}. According to the description given there, this collection of 102 TSP instances was provided by Andre Rohe, based on VLSI data sets studied at the Forschungsinstitut für Diskrete Mathematik, Universität Bonn. 

The name of all instances contains their size.
To save some space and properly layout the table, the absolute number of possible edges remaining after the diamond test is not given, instead column DT shows the factor w.r.t. instance size.
Each instance was solved ten times; mean values of the runtime of each part of the algorithm are given in the corresponding columns.
Standard deviations were small; they are left out for layout reasons.

The LMT-heuristic was sufficient to solve the MWT for all instances, except \texttt{lra498378}, which contained two non-simple polygonal faces.

\begin{landscape}
	\begingroup \small %
\bgroup \small \setlength {\tabcolsep }{4pt}\begin {longtabu}{c|rrr|rr|rrrrrr}%
\caption {VLSI statistics}\\ \toprule & \multicolumn {3}{c|}{Possible edges after} & \multicolumn {2}{c|}{Certain edges} & \multicolumn {6}{c}{Time in ms}\\Instance&DT&LMT&LMT+&LMT&LMT+&Total&DT&LMT-Init&LMT&LMT+&Dyn. Prog.\\\toprule \endhead \midrule \multicolumn {12}{r}{{Continued on next page}} \\ \bottomrule \endfoot \midrule \multicolumn {12}{r}{{End}} \\ \bottomrule \endlastfoot %
ara238025&\pgfutilensuremath {12.44}&\pgfutilensuremath {403{,}634}&\pgfutilensuremath {351{,}298}&\pgfutilensuremath {578{,}814}&\pgfutilensuremath {581{,}677}&\pgfutilensuremath {15{,}325.50}&\pgfutilensuremath {4{,}954.62}&\pgfutilensuremath {446.95}&\pgfutilensuremath {496.27}&\pgfutilensuremath {9{,}279.19}&\pgfutilensuremath {148.00}\\%
bby34656&\pgfutilensuremath {12.32}&\pgfutilensuremath {56{,}775}&\pgfutilensuremath {50{,}720}&\pgfutilensuremath {84{,}829}&\pgfutilensuremath {85{,}194}&\pgfutilensuremath {318.10}&\pgfutilensuremath {113.57}&\pgfutilensuremath {65.39}&\pgfutilensuremath {52.82}&\pgfutilensuremath {65.47}&\pgfutilensuremath {20.79}\\%
bbz25234&\pgfutilensuremath {12.46}&\pgfutilensuremath {40{,}912}&\pgfutilensuremath {36{,}697}&\pgfutilensuremath {61{,}588}&\pgfutilensuremath {61{,}762}&\pgfutilensuremath {240.44}&\pgfutilensuremath {85.97}&\pgfutilensuremath {47.97}&\pgfutilensuremath {40.89}&\pgfutilensuremath {49.70}&\pgfutilensuremath {15.86}\\%
bch2762&\pgfutilensuremath {12.39}&\pgfutilensuremath {5{,}434}&\pgfutilensuremath {4{,}733}&\pgfutilensuremath {6{,}452}&\pgfutilensuremath {6{,}467}&\pgfutilensuremath {29.61}&\pgfutilensuremath {9.00}&\pgfutilensuremath {4.75}&\pgfutilensuremath {6.56}&\pgfutilensuremath {7.84}&\pgfutilensuremath {1.46}\\%
bck2217&\pgfutilensuremath {11.7}&\pgfutilensuremath {4{,}138}&\pgfutilensuremath {3{,}545}&\pgfutilensuremath {5{,}333}&\pgfutilensuremath {5{,}354}&\pgfutilensuremath {23.40}&\pgfutilensuremath {6.98}&\pgfutilensuremath {3.61}&\pgfutilensuremath {4.92}&\pgfutilensuremath {6.74}&\pgfutilensuremath {1.15}\\%
bcl380&\pgfutilensuremath {11.38}&\pgfutilensuremath {756}&\pgfutilensuremath {683}&\pgfutilensuremath {854}&\pgfutilensuremath {854}&\pgfutilensuremath {4.96}&\pgfutilensuremath {2.01}&\pgfutilensuremath {0.59}&\pgfutilensuremath {1.47}&\pgfutilensuremath {0.66}&\pgfutilensuremath {0.23}\\%
beg3293&\pgfutilensuremath {12.83}&\pgfutilensuremath {7{,}376}&\pgfutilensuremath {6{,}308}&\pgfutilensuremath {7{,}677}&\pgfutilensuremath {7{,}708}&\pgfutilensuremath {41.56}&\pgfutilensuremath {11.61}&\pgfutilensuremath {5.79}&\pgfutilensuremath {8.72}&\pgfutilensuremath {13.75}&\pgfutilensuremath {1.69}\\%
bgb4355&\pgfutilensuremath {11.59}&\pgfutilensuremath {6{,}517}&\pgfutilensuremath {5{,}713}&\pgfutilensuremath {10{,}708}&\pgfutilensuremath {10{,}754}&\pgfutilensuremath {36.06}&\pgfutilensuremath {13.91}&\pgfutilensuremath {6.98}&\pgfutilensuremath {6.54}&\pgfutilensuremath {6.40}&\pgfutilensuremath {2.23}\\%
bgd4396&\pgfutilensuremath {12.86}&\pgfutilensuremath {8{,}891}&\pgfutilensuremath {7{,}861}&\pgfutilensuremath {10{,}400}&\pgfutilensuremath {10{,}433}&\pgfutilensuremath {46.26}&\pgfutilensuremath {15.42}&\pgfutilensuremath {7.87}&\pgfutilensuremath {8.64}&\pgfutilensuremath {12.01}&\pgfutilensuremath {2.31}\\%
bgf4475&\pgfutilensuremath {12.49}&\pgfutilensuremath {8{,}816}&\pgfutilensuremath {7{,}868}&\pgfutilensuremath {10{,}530}&\pgfutilensuremath {10{,}563}&\pgfutilensuremath {45.50}&\pgfutilensuremath {16.02}&\pgfutilensuremath {7.72}&\pgfutilensuremath {8.03}&\pgfutilensuremath {11.36}&\pgfutilensuremath {2.35}\\%
bna56769&\pgfutilensuremath {12.35}&\pgfutilensuremath {91{,}991}&\pgfutilensuremath {82{,}706}&\pgfutilensuremath {139{,}371}&\pgfutilensuremath {139{,}938}&\pgfutilensuremath {512.02}&\pgfutilensuremath {180.09}&\pgfutilensuremath {107.54}&\pgfutilensuremath {82.92}&\pgfutilensuremath {106.94}&\pgfutilensuremath {34.43}\\%
bnd7168&\pgfutilensuremath {12.04}&\pgfutilensuremath {12{,}274}&\pgfutilensuremath {10{,}781}&\pgfutilensuremath {17{,}377}&\pgfutilensuremath {17{,}434}&\pgfutilensuremath {66.48}&\pgfutilensuremath {23.19}&\pgfutilensuremath {12.64}&\pgfutilensuremath {11.08}&\pgfutilensuremath {15.56}&\pgfutilensuremath {3.99}\\%
boa28924&\pgfutilensuremath {12.13}&\pgfutilensuremath {48{,}935}&\pgfutilensuremath {44{,}347}&\pgfutilensuremath {69{,}974}&\pgfutilensuremath {70{,}154}&\pgfutilensuremath {271.97}&\pgfutilensuremath {96.98}&\pgfutilensuremath {54.02}&\pgfutilensuremath {45.72}&\pgfutilensuremath {57.37}&\pgfutilensuremath {17.83}\\%
bva2144&\pgfutilensuremath {11.5}&\pgfutilensuremath {3{,}619}&\pgfutilensuremath {3{,}093}&\pgfutilensuremath {5{,}200}&\pgfutilensuremath {5{,}216}&\pgfutilensuremath {20.92}&\pgfutilensuremath {6.70}&\pgfutilensuremath {3.41}&\pgfutilensuremath {4.48}&\pgfutilensuremath {5.27}&\pgfutilensuremath {1.05}\\%
dan59296&\pgfutilensuremath {12.5}&\pgfutilensuremath {100{,}280}&\pgfutilensuremath {88{,}563}&\pgfutilensuremath {144{,}508}&\pgfutilensuremath {145{,}112}&\pgfutilensuremath {570.98}&\pgfutilensuremath {191.87}&\pgfutilensuremath {115.42}&\pgfutilensuremath {90.45}&\pgfutilensuremath {136.86}&\pgfutilensuremath {36.28}\\%
dbj2924&\pgfutilensuremath {11.53}&\pgfutilensuremath {4{,}970}&\pgfutilensuremath {4{,}271}&\pgfutilensuremath {6{,}999}&\pgfutilensuremath {7{,}033}&\pgfutilensuremath {27.15}&\pgfutilensuremath {8.91}&\pgfutilensuremath {4.65}&\pgfutilensuremath {6.14}&\pgfutilensuremath {6.10}&\pgfutilensuremath {1.34}\\%
dca1389&\pgfutilensuremath {11.45}&\pgfutilensuremath {2{,}104}&\pgfutilensuremath {1{,}817}&\pgfutilensuremath {3{,}388}&\pgfutilensuremath {3{,}400}&\pgfutilensuremath {14.17}&\pgfutilensuremath {5.02}&\pgfutilensuremath {2.20}&\pgfutilensuremath {3.68}&\pgfutilensuremath {2.56}&\pgfutilensuremath {0.70}\\%
dcb2086&\pgfutilensuremath {12.36}&\pgfutilensuremath {4{,}001}&\pgfutilensuremath {3{,}563}&\pgfutilensuremath {4{,}924}&\pgfutilensuremath {4{,}941}&\pgfutilensuremath {21.76}&\pgfutilensuremath {6.85}&\pgfutilensuremath {3.53}&\pgfutilensuremath {5.38}&\pgfutilensuremath {4.93}&\pgfutilensuremath {1.06}\\%
dcc1911&\pgfutilensuremath {11.35}&\pgfutilensuremath {2{,}718}&\pgfutilensuremath {2{,}480}&\pgfutilensuremath {4{,}701}&\pgfutilensuremath {4{,}716}&\pgfutilensuremath {17.54}&\pgfutilensuremath {6.15}&\pgfutilensuremath {2.99}&\pgfutilensuremath {5.43}&\pgfutilensuremath {2.01}&\pgfutilensuremath {0.96}\\%
dea2382&\pgfutilensuremath {12.31}&\pgfutilensuremath {5{,}182}&\pgfutilensuremath {4{,}326}&\pgfutilensuremath {5{,}521}&\pgfutilensuremath {5{,}560}&\pgfutilensuremath {30.76}&\pgfutilensuremath {8.60}&\pgfutilensuremath {4.02}&\pgfutilensuremath {6.60}&\pgfutilensuremath {10.32}&\pgfutilensuremath {1.22}\\%
dga9698&\pgfutilensuremath {13.01}&\pgfutilensuremath {21{,}943}&\pgfutilensuremath {18{,}586}&\pgfutilensuremath {22{,}603}&\pgfutilensuremath {22{,}700}&\pgfutilensuremath {122.77}&\pgfutilensuremath {33.76}&\pgfutilensuremath {18.84}&\pgfutilensuremath {18.03}&\pgfutilensuremath {46.50}&\pgfutilensuremath {5.62}\\%
dhb3386&\pgfutilensuremath {11.97}&\pgfutilensuremath {6{,}326}&\pgfutilensuremath {5{,}471}&\pgfutilensuremath {8{,}089}&\pgfutilensuremath {8{,}119}&\pgfutilensuremath {34.69}&\pgfutilensuremath {10.73}&\pgfutilensuremath {5.56}&\pgfutilensuremath {7.35}&\pgfutilensuremath {9.37}&\pgfutilensuremath {1.66}\\%
dja1436&\pgfutilensuremath {11.6}&\pgfutilensuremath {2{,}472}&\pgfutilensuremath {2{,}084}&\pgfutilensuremath {3{,}474}&\pgfutilensuremath {3{,}486}&\pgfutilensuremath {16.39}&\pgfutilensuremath {4.66}&\pgfutilensuremath {2.28}&\pgfutilensuremath {4.09}&\pgfutilensuremath {4.63}&\pgfutilensuremath {0.73}\\%
djb2036&\pgfutilensuremath {12.23}&\pgfutilensuremath {3{,}810}&\pgfutilensuremath {3{,}185}&\pgfutilensuremath {4{,}892}&\pgfutilensuremath {4{,}916}&\pgfutilensuremath {23.32}&\pgfutilensuremath {6.84}&\pgfutilensuremath {3.44}&\pgfutilensuremath {5.00}&\pgfutilensuremath {6.98}&\pgfutilensuremath {1.06}\\%
djc1785&\pgfutilensuremath {11.27}&\pgfutilensuremath {2{,}969}&\pgfutilensuremath {2{,}545}&\pgfutilensuremath {4{,}375}&\pgfutilensuremath {4{,}399}&\pgfutilensuremath {18.42}&\pgfutilensuremath {5.70}&\pgfutilensuremath {2.77}&\pgfutilensuremath {4.70}&\pgfutilensuremath {4.32}&\pgfutilensuremath {0.92}\\%
dka1376&\pgfutilensuremath {11.94}&\pgfutilensuremath {2{,}420}&\pgfutilensuremath {2{,}090}&\pgfutilensuremath {3{,}304}&\pgfutilensuremath {3{,}314}&\pgfutilensuremath {14.50}&\pgfutilensuremath {4.45}&\pgfutilensuremath {2.28}&\pgfutilensuremath {4.44}&\pgfutilensuremath {2.63}&\pgfutilensuremath {0.69}\\%
dkc3938&\pgfutilensuremath {11.99}&\pgfutilensuremath {6{,}612}&\pgfutilensuremath {5{,}644}&\pgfutilensuremath {9{,}603}&\pgfutilensuremath {9{,}642}&\pgfutilensuremath {41.12}&\pgfutilensuremath {13.39}&\pgfutilensuremath {6.46}&\pgfutilensuremath {8.69}&\pgfutilensuremath {10.74}&\pgfutilensuremath {1.83}\\%
dkd1973&\pgfutilensuremath {10.6}&\pgfutilensuremath {1{,}987}&\pgfutilensuremath {1{,}824}&\pgfutilensuremath {5{,}054}&\pgfutilensuremath {5{,}059}&\pgfutilensuremath {16.54}&\pgfutilensuremath {6.41}&\pgfutilensuremath {2.82}&\pgfutilensuremath {5.05}&\pgfutilensuremath {1.38}&\pgfutilensuremath {0.87}\\%
dke3097&\pgfutilensuremath {12.21}&\pgfutilensuremath {5{,}831}&\pgfutilensuremath {4{,}978}&\pgfutilensuremath {7{,}401}&\pgfutilensuremath {7{,}424}&\pgfutilensuremath {33.95}&\pgfutilensuremath {10.44}&\pgfutilensuremath {5.18}&\pgfutilensuremath {6.65}&\pgfutilensuremath {10.16}&\pgfutilensuremath {1.52}\\%
dkf3954&\pgfutilensuremath {11.78}&\pgfutilensuremath {6{,}204}&\pgfutilensuremath {5{,}515}&\pgfutilensuremath {9{,}708}&\pgfutilensuremath {9{,}742}&\pgfutilensuremath {35.14}&\pgfutilensuremath {12.58}&\pgfutilensuremath {6.43}&\pgfutilensuremath {8.02}&\pgfutilensuremath {6.21}&\pgfutilensuremath {1.90}\\%
dkg813&\pgfutilensuremath {11.36}&\pgfutilensuremath {1{,}068}&\pgfutilensuremath {968}&\pgfutilensuremath {2{,}019}&\pgfutilensuremath {2{,}024}&\pgfutilensuremath {8.13}&\pgfutilensuremath {2.96}&\pgfutilensuremath {1.27}&\pgfutilensuremath {2.74}&\pgfutilensuremath {0.73}&\pgfutilensuremath {0.42}\\%
dlb3694&\pgfutilensuremath {10.98}&\pgfutilensuremath {4{,}519}&\pgfutilensuremath {4{,}220}&\pgfutilensuremath {9{,}167}&\pgfutilensuremath {9{,}174}&\pgfutilensuremath {29.09}&\pgfutilensuremath {11.51}&\pgfutilensuremath {6.13}&\pgfutilensuremath {6.98}&\pgfutilensuremath {2.50}&\pgfutilensuremath {1.96}\\%
fdp3256&\pgfutilensuremath {11.73}&\pgfutilensuremath {5{,}064}&\pgfutilensuremath {4{,}617}&\pgfutilensuremath {7{,}774}&\pgfutilensuremath {7{,}794}&\pgfutilensuremath {28.53}&\pgfutilensuremath {10.43}&\pgfutilensuremath {5.26}&\pgfutilensuremath {6.96}&\pgfutilensuremath {4.14}&\pgfutilensuremath {1.74}\\%
fea5557&\pgfutilensuremath {11.67}&\pgfutilensuremath {7{,}663}&\pgfutilensuremath {7{,}012}&\pgfutilensuremath {13{,}568}&\pgfutilensuremath {13{,}591}&\pgfutilensuremath {45.90}&\pgfutilensuremath {18.49}&\pgfutilensuremath {9.08}&\pgfutilensuremath {8.83}&\pgfutilensuremath {6.45}&\pgfutilensuremath {3.05}\\%
fht47608&\pgfutilensuremath {12.63}&\pgfutilensuremath {74{,}464}&\pgfutilensuremath {67{,}605}&\pgfutilensuremath {117{,}543}&\pgfutilensuremath {117{,}852}&\pgfutilensuremath {437.76}&\pgfutilensuremath {165.02}&\pgfutilensuremath {91.31}&\pgfutilensuremath {74.70}&\pgfutilensuremath {76.51}&\pgfutilensuremath {30.14}\\%
fjr3672&\pgfutilensuremath {9}&\pgfutilensuremath {3{,}441}&\pgfutilensuremath {3{,}303}&\pgfutilensuremath {9{,}476}&\pgfutilensuremath {9{,}480}&\pgfutilensuremath {27.06}&\pgfutilensuremath {12.90}&\pgfutilensuremath {4.41}&\pgfutilensuremath {6.01}&\pgfutilensuremath {1.94}&\pgfutilensuremath {1.80}\\%
fjs3649&\pgfutilensuremath {8.88}&\pgfutilensuremath {3{,}039}&\pgfutilensuremath {2{,}892}&\pgfutilensuremath {9{,}473}&\pgfutilensuremath {9{,}480}&\pgfutilensuremath {25.26}&\pgfutilensuremath {12.17}&\pgfutilensuremath {4.31}&\pgfutilensuremath {5.83}&\pgfutilensuremath {1.20}&\pgfutilensuremath {1.75}\\%
fma21553&\pgfutilensuremath {12.13}&\pgfutilensuremath {30{,}434}&\pgfutilensuremath {26{,}810}&\pgfutilensuremath {53{,}500}&\pgfutilensuremath {53{,}779}&\pgfutilensuremath {193.02}&\pgfutilensuremath {74.42}&\pgfutilensuremath {38.76}&\pgfutilensuremath {34.11}&\pgfutilensuremath {32.60}&\pgfutilensuremath {13.10}\\%
fna52057&\pgfutilensuremath {12.34}&\pgfutilensuremath {77{,}235}&\pgfutilensuremath {69{,}688}&\pgfutilensuremath {127{,}955}&\pgfutilensuremath {128{,}384}&\pgfutilensuremath {466.54}&\pgfutilensuremath {177.26}&\pgfutilensuremath {97.45}&\pgfutilensuremath {77.89}&\pgfutilensuremath {81.54}&\pgfutilensuremath {32.31}\\%
fnb1615&\pgfutilensuremath {12.29}&\pgfutilensuremath {2{,}688}&\pgfutilensuremath {2{,}425}&\pgfutilensuremath {4{,}004}&\pgfutilensuremath {4{,}011}&\pgfutilensuremath {17.69}&\pgfutilensuremath {5.86}&\pgfutilensuremath {2.76}&\pgfutilensuremath {4.88}&\pgfutilensuremath {3.27}&\pgfutilensuremath {0.92}\\%
fnc19402&\pgfutilensuremath {12.18}&\pgfutilensuremath {31{,}461}&\pgfutilensuremath {27{,}808}&\pgfutilensuremath {47{,}248}&\pgfutilensuremath {47{,}476}&\pgfutilensuremath {179.74}&\pgfutilensuremath {64.15}&\pgfutilensuremath {35.00}&\pgfutilensuremath {30.50}&\pgfutilensuremath {37.90}&\pgfutilensuremath {12.16}\\%
fqm5087&\pgfutilensuremath {11.41}&\pgfutilensuremath {6{,}126}&\pgfutilensuremath {5{,}818}&\pgfutilensuremath {12{,}813}&\pgfutilensuremath {12{,}829}&\pgfutilensuremath {41.30}&\pgfutilensuremath {18.92}&\pgfutilensuremath {8.13}&\pgfutilensuremath {7.90}&\pgfutilensuremath {3.58}&\pgfutilensuremath {2.77}\\%
fra1488&\pgfutilensuremath {10.83}&\pgfutilensuremath {1{,}832}&\pgfutilensuremath {1{,}752}&\pgfutilensuremath {3{,}674}&\pgfutilensuremath {3{,}680}&\pgfutilensuremath {13.26}&\pgfutilensuremath {5.21}&\pgfutilensuremath {2.20}&\pgfutilensuremath {3.98}&\pgfutilensuremath {1.00}&\pgfutilensuremath {0.86}\\%
frh19289&\pgfutilensuremath {12.18}&\pgfutilensuremath {30{,}991}&\pgfutilensuremath {27{,}494}&\pgfutilensuremath {46{,}901}&\pgfutilensuremath {47{,}047}&\pgfutilensuremath {176.38}&\pgfutilensuremath {62.45}&\pgfutilensuremath {35.20}&\pgfutilensuremath {30.29}&\pgfutilensuremath {36.82}&\pgfutilensuremath {11.59}\\%
frv4410&\pgfutilensuremath {11.51}&\pgfutilensuremath {6{,}122}&\pgfutilensuremath {5{,}494}&\pgfutilensuremath {10{,}967}&\pgfutilensuremath {10{,}987}&\pgfutilensuremath {38.18}&\pgfutilensuremath {15.13}&\pgfutilensuremath {7.01}&\pgfutilensuremath {6.34}&\pgfutilensuremath {7.27}&\pgfutilensuremath {2.42}\\%
fry33203&\pgfutilensuremath {12.23}&\pgfutilensuremath {49{,}364}&\pgfutilensuremath {44{,}365}&\pgfutilensuremath {81{,}285}&\pgfutilensuremath {81{,}586}&\pgfutilensuremath {299.88}&\pgfutilensuremath {111.15}&\pgfutilensuremath {61.90}&\pgfutilensuremath {51.71}&\pgfutilensuremath {54.22}&\pgfutilensuremath {20.84}\\%
fyg28534&\pgfutilensuremath {12.7}&\pgfutilensuremath {58{,}992}&\pgfutilensuremath {51{,}953}&\pgfutilensuremath {66{,}600}&\pgfutilensuremath {66{,}812}&\pgfutilensuremath {299.51}&\pgfutilensuremath {94.29}&\pgfutilensuremath {56.44}&\pgfutilensuremath {46.84}&\pgfutilensuremath {84.61}&\pgfutilensuremath {17.29}\\%
ics39603&\pgfutilensuremath {12.17}&\pgfutilensuremath {70{,}351}&\pgfutilensuremath {62{,}412}&\pgfutilensuremath {96{,}188}&\pgfutilensuremath {96{,}379}&\pgfutilensuremath {392.97}&\pgfutilensuremath {129.38}&\pgfutilensuremath {73.01}&\pgfutilensuremath {59.62}&\pgfutilensuremath {106.60}&\pgfutilensuremath {24.30}\\%
icw1483&\pgfutilensuremath {11.4}&\pgfutilensuremath {2{,}476}&\pgfutilensuremath {2{,}210}&\pgfutilensuremath {3{,}572}&\pgfutilensuremath {3{,}581}&\pgfutilensuremath {24.66}&\pgfutilensuremath {14.19}&\pgfutilensuremath {2.46}&\pgfutilensuremath {4.39}&\pgfutilensuremath {2.78}&\pgfutilensuremath {0.84}\\%
icx28698&\pgfutilensuremath {13.31}&\pgfutilensuremath {67{,}496}&\pgfutilensuremath {60{,}492}&\pgfutilensuremath {67{,}761}&\pgfutilensuremath {67{,}945}&\pgfutilensuremath {346.87}&\pgfutilensuremath {97.32}&\pgfutilensuremath {59.60}&\pgfutilensuremath {51.08}&\pgfutilensuremath {120.27}&\pgfutilensuremath {18.56}\\%
ida8197&\pgfutilensuremath {11.94}&\pgfutilensuremath {12{,}449}&\pgfutilensuremath {11{,}251}&\pgfutilensuremath {20{,}075}&\pgfutilensuremath {20{,}129}&\pgfutilensuremath {74.08}&\pgfutilensuremath {26.88}&\pgfutilensuremath {14.32}&\pgfutilensuremath {12.92}&\pgfutilensuremath {14.96}&\pgfutilensuremath {4.97}\\%
ido21215&\pgfutilensuremath {12.28}&\pgfutilensuremath {32{,}108}&\pgfutilensuremath {28{,}264}&\pgfutilensuremath {52{,}421}&\pgfutilensuremath {52{,}546}&\pgfutilensuremath {202.58}&\pgfutilensuremath {68.81}&\pgfutilensuremath {39.19}&\pgfutilensuremath {34.21}&\pgfutilensuremath {47.49}&\pgfutilensuremath {12.86}\\%
ird29514&\pgfutilensuremath {13.08}&\pgfutilensuremath {62{,}579}&\pgfutilensuremath {54{,}505}&\pgfutilensuremath {70{,}537}&\pgfutilensuremath {70{,}807}&\pgfutilensuremath {348.57}&\pgfutilensuremath {96.42}&\pgfutilensuremath {58.83}&\pgfutilensuremath {50.09}&\pgfutilensuremath {124.90}&\pgfutilensuremath {18.28}\\%
irw2802&\pgfutilensuremath {12.54}&\pgfutilensuremath {5{,}053}&\pgfutilensuremath {4{,}330}&\pgfutilensuremath {6{,}705}&\pgfutilensuremath {6{,}725}&\pgfutilensuremath {29.62}&\pgfutilensuremath {9.77}&\pgfutilensuremath {4.82}&\pgfutilensuremath {6.46}&\pgfutilensuremath {7.20}&\pgfutilensuremath {1.37}\\%
irx28268&\pgfutilensuremath {12.21}&\pgfutilensuremath {36{,}680}&\pgfutilensuremath {33{,}430}&\pgfutilensuremath {71{,}115}&\pgfutilensuremath {71{,}250}&\pgfutilensuremath {243.50}&\pgfutilensuremath {92.33}&\pgfutilensuremath {52.91}&\pgfutilensuremath {41.81}&\pgfutilensuremath {39.36}&\pgfutilensuremath {17.04}\\%
lap7454&\pgfutilensuremath {12.55}&\pgfutilensuremath {12{,}993}&\pgfutilensuremath {11{,}518}&\pgfutilensuremath {17{,}751}&\pgfutilensuremath {17{,}816}&\pgfutilensuremath {70.70}&\pgfutilensuremath {24.98}&\pgfutilensuremath {13.59}&\pgfutilensuremath {12.40}&\pgfutilensuremath {15.45}&\pgfutilensuremath {4.27}\\%
ley2323&\pgfutilensuremath {13.32}&\pgfutilensuremath {5{,}550}&\pgfutilensuremath {4{,}866}&\pgfutilensuremath {5{,}081}&\pgfutilensuremath {5{,}098}&\pgfutilensuremath {27.44}&\pgfutilensuremath {8.72}&\pgfutilensuremath {4.18}&\pgfutilensuremath {6.96}&\pgfutilensuremath {6.41}&\pgfutilensuremath {1.15}\\%
lim963&\pgfutilensuremath {11.18}&\pgfutilensuremath {1{,}743}&\pgfutilensuremath {1{,}506}&\pgfutilensuremath {2{,}280}&\pgfutilensuremath {2{,}290}&\pgfutilensuremath {9.86}&\pgfutilensuremath {3.04}&\pgfutilensuremath {1.47}&\pgfutilensuremath {2.76}&\pgfutilensuremath {1.95}&\pgfutilensuremath {0.63}\\%
lra498378&\pgfutilensuremath {16.03}&\pgfutilensuremath {1{,}214{,}570}&\pgfutilensuremath {1{,}061{,}438}&\pgfutilensuremath {1{,}204{,}368}&\pgfutilensuremath {1{,}209{,}111}&\pgfutilensuremath {382{,}932.00}&\pgfutilensuremath {44{,}267.10}&\pgfutilensuremath {1{,}238.36}&\pgfutilensuremath {7{,}532.35}&\pgfutilensuremath {329{,}292.00}&\pgfutilensuremath {599.16}\\%
lrb744710&\pgfutilensuremath {12.02}&\pgfutilensuremath {1{,}233{,}436}&\pgfutilensuremath {1{,}073{,}199}&\pgfutilensuremath {1{,}872{,}986}&\pgfutilensuremath {1{,}879{,}647}&\pgfutilensuremath {484{,}430.00}&\pgfutilensuremath {7{,}952.68}&\pgfutilensuremath {1{,}377.69}&\pgfutilensuremath {2{,}661.36}&\pgfutilensuremath {471{,}564.00}&\pgfutilensuremath {872.17}\\%
lsb22777&\pgfutilensuremath {12.55}&\pgfutilensuremath {30{,}416}&\pgfutilensuremath {27{,}355}&\pgfutilensuremath {57{,}192}&\pgfutilensuremath {57{,}326}&\pgfutilensuremath {206.72}&\pgfutilensuremath {76.85}&\pgfutilensuremath {43.53}&\pgfutilensuremath {36.48}&\pgfutilensuremath {35.91}&\pgfutilensuremath {13.91}\\%
lsm2854&\pgfutilensuremath {11.7}&\pgfutilensuremath {4{,}297}&\pgfutilensuremath {3{,}865}&\pgfutilensuremath {7{,}044}&\pgfutilensuremath {7{,}062}&\pgfutilensuremath {25.38}&\pgfutilensuremath {9.27}&\pgfutilensuremath {4.63}&\pgfutilensuremath {5.92}&\pgfutilensuremath {4.07}&\pgfutilensuremath {1.49}\\%
lsn3119&\pgfutilensuremath {11.64}&\pgfutilensuremath {5{,}203}&\pgfutilensuremath {4{,}566}&\pgfutilensuremath {7{,}537}&\pgfutilensuremath {7{,}567}&\pgfutilensuremath {28.84}&\pgfutilensuremath {9.86}&\pgfutilensuremath {4.99}&\pgfutilensuremath {6.58}&\pgfutilensuremath {5.75}&\pgfutilensuremath {1.66}\\%
lta3140&\pgfutilensuremath {12.65}&\pgfutilensuremath {5{,}329}&\pgfutilensuremath {4{,}815}&\pgfutilensuremath {7{,}620}&\pgfutilensuremath {7{,}645}&\pgfutilensuremath {31.43}&\pgfutilensuremath {10.54}&\pgfutilensuremath {5.49}&\pgfutilensuremath {8.17}&\pgfutilensuremath {5.57}&\pgfutilensuremath {1.66}\\%
ltb3729&\pgfutilensuremath {11.59}&\pgfutilensuremath {5{,}862}&\pgfutilensuremath {5{,}138}&\pgfutilensuremath {9{,}073}&\pgfutilensuremath {9{,}111}&\pgfutilensuremath {34.51}&\pgfutilensuremath {12.24}&\pgfutilensuremath {5.90}&\pgfutilensuremath {7.89}&\pgfutilensuremath {6.70}&\pgfutilensuremath {1.78}\\%
mlt2597&\pgfutilensuremath {12.15}&\pgfutilensuremath {5{,}025}&\pgfutilensuremath {4{,}440}&\pgfutilensuremath {6{,}120}&\pgfutilensuremath {6{,}138}&\pgfutilensuremath {27.34}&\pgfutilensuremath {8.60}&\pgfutilensuremath {4.36}&\pgfutilensuremath {6.13}&\pgfutilensuremath {6.92}&\pgfutilensuremath {1.33}\\%
pba38478&\pgfutilensuremath {12.66}&\pgfutilensuremath {68{,}443}&\pgfutilensuremath {60{,}780}&\pgfutilensuremath {93{,}518}&\pgfutilensuremath {93{,}779}&\pgfutilensuremath {393.44}&\pgfutilensuremath {129.71}&\pgfutilensuremath {73.83}&\pgfutilensuremath {60.54}&\pgfutilensuremath {105.28}&\pgfutilensuremath {24.01}\\%
pbd984&\pgfutilensuremath {10.77}&\pgfutilensuremath {1{,}230}&\pgfutilensuremath {1{,}126}&\pgfutilensuremath {2{,}467}&\pgfutilensuremath {2{,}472}&\pgfutilensuremath {8.97}&\pgfutilensuremath {3.00}&\pgfutilensuremath {1.47}&\pgfutilensuremath {2.99}&\pgfutilensuremath {0.94}&\pgfutilensuremath {0.57}\\%
pbh30440&\pgfutilensuremath {12.69}&\pgfutilensuremath {50{,}742}&\pgfutilensuremath {46{,}520}&\pgfutilensuremath {73{,}910}&\pgfutilensuremath {74{,}078}&\pgfutilensuremath {281.17}&\pgfutilensuremath {100.29}&\pgfutilensuremath {58.66}&\pgfutilensuremath {51.91}&\pgfutilensuremath {52.26}&\pgfutilensuremath {17.99}\\%
pbk411&\pgfutilensuremath {11.29}&\pgfutilensuremath {482}&\pgfutilensuremath {450}&\pgfutilensuremath {1{,}013}&\pgfutilensuremath {1{,}014}&\pgfutilensuremath {3.78}&\pgfutilensuremath {1.34}&\pgfutilensuremath {0.65}&\pgfutilensuremath {1.34}&\pgfutilensuremath {0.21}&\pgfutilensuremath {0.25}\\%
pbl395&\pgfutilensuremath {10.86}&\pgfutilensuremath {427}&\pgfutilensuremath {389}&\pgfutilensuremath {997}&\pgfutilensuremath {1{,}000}&\pgfutilensuremath {4.40}&\pgfutilensuremath {2.06}&\pgfutilensuremath {0.58}&\pgfutilensuremath {1.22}&\pgfutilensuremath {0.31}&\pgfutilensuremath {0.23}\\%
pbm436&\pgfutilensuremath {11.41}&\pgfutilensuremath {543}&\pgfutilensuremath {491}&\pgfutilensuremath {1{,}081}&\pgfutilensuremath {1{,}083}&\pgfutilensuremath {4.15}&\pgfutilensuremath {1.34}&\pgfutilensuremath {0.69}&\pgfutilensuremath {1.43}&\pgfutilensuremath {0.44}&\pgfutilensuremath {0.25}\\%
pbn423&\pgfutilensuremath {11.26}&\pgfutilensuremath {538}&\pgfutilensuremath {494}&\pgfutilensuremath {1{,}036}&\pgfutilensuremath {1{,}038}&\pgfutilensuremath {3.91}&\pgfutilensuremath {1.36}&\pgfutilensuremath {0.68}&\pgfutilensuremath {1.28}&\pgfutilensuremath {0.35}&\pgfutilensuremath {0.24}\\%
pds2566&\pgfutilensuremath {12.31}&\pgfutilensuremath {4{,}431}&\pgfutilensuremath {3{,}937}&\pgfutilensuremath {6{,}132}&\pgfutilensuremath {6{,}143}&\pgfutilensuremath {25.77}&\pgfutilensuremath {8.46}&\pgfutilensuremath {4.40}&\pgfutilensuremath {6.40}&\pgfutilensuremath {5.30}&\pgfutilensuremath {1.21}\\%
pia3056&\pgfutilensuremath {11.78}&\pgfutilensuremath {4{,}732}&\pgfutilensuremath {4{,}342}&\pgfutilensuremath {7{,}539}&\pgfutilensuremath {7{,}551}&\pgfutilensuremath {26.54}&\pgfutilensuremath {9.15}&\pgfutilensuremath {5.01}&\pgfutilensuremath {6.22}&\pgfutilensuremath {4.56}&\pgfutilensuremath {1.59}\\%
pjh17845&\pgfutilensuremath {12.31}&\pgfutilensuremath {24{,}219}&\pgfutilensuremath {22{,}140}&\pgfutilensuremath {44{,}814}&\pgfutilensuremath {44{,}894}&\pgfutilensuremath {153.80}&\pgfutilensuremath {57.02}&\pgfutilensuremath {33.19}&\pgfutilensuremath {28.21}&\pgfutilensuremath {24.68}&\pgfutilensuremath {10.68}\\%
pka379&\pgfutilensuremath {9.53}&\pgfutilensuremath {317}&\pgfutilensuremath {309}&\pgfutilensuremath {977}&\pgfutilensuremath {977}&\pgfutilensuremath {3.48}&\pgfutilensuremath {1.77}&\pgfutilensuremath {0.48}&\pgfutilensuremath {0.93}&\pgfutilensuremath {0.10}&\pgfutilensuremath {0.21}\\%
pma343&\pgfutilensuremath {9.29}&\pgfutilensuremath {291}&\pgfutilensuremath {273}&\pgfutilensuremath {905}&\pgfutilensuremath {905}&\pgfutilensuremath {3.22}&\pgfutilensuremath {1.58}&\pgfutilensuremath {0.43}&\pgfutilensuremath {0.81}&\pgfutilensuremath {0.22}&\pgfutilensuremath {0.18}\\%
rbu737&\pgfutilensuremath {9.96}&\pgfutilensuremath {1{,}270}&\pgfutilensuremath {1{,}114}&\pgfutilensuremath {1{,}751}&\pgfutilensuremath {1{,}757}&\pgfutilensuremath {6.48}&\pgfutilensuremath {2.26}&\pgfutilensuremath {1.01}&\pgfutilensuremath {1.69}&\pgfutilensuremath {1.10}&\pgfutilensuremath {0.42}\\%
rbv1583&\pgfutilensuremath {11.34}&\pgfutilensuremath {1{,}937}&\pgfutilensuremath {1{,}834}&\pgfutilensuremath {3{,}920}&\pgfutilensuremath {3{,}926}&\pgfutilensuremath {14.11}&\pgfutilensuremath {5.49}&\pgfutilensuremath {2.48}&\pgfutilensuremath {4.30}&\pgfutilensuremath {0.94}&\pgfutilensuremath {0.90}\\%
rbw2481&\pgfutilensuremath {12.3}&\pgfutilensuremath {4{,}500}&\pgfutilensuremath {4{,}007}&\pgfutilensuremath {5{,}881}&\pgfutilensuremath {5{,}900}&\pgfutilensuremath {26.73}&\pgfutilensuremath {8.56}&\pgfutilensuremath {4.18}&\pgfutilensuremath {6.28}&\pgfutilensuremath {6.45}&\pgfutilensuremath {1.25}\\%
rbx711&\pgfutilensuremath {10.19}&\pgfutilensuremath {936}&\pgfutilensuremath {872}&\pgfutilensuremath {1{,}698}&\pgfutilensuremath {1{,}701}&\pgfutilensuremath {5.57}&\pgfutilensuremath {2.02}&\pgfutilensuremath {0.99}&\pgfutilensuremath {1.62}&\pgfutilensuremath {0.51}&\pgfutilensuremath {0.44}\\%
rby1599&\pgfutilensuremath {11.25}&\pgfutilensuremath {2{,}006}&\pgfutilensuremath {1{,}879}&\pgfutilensuremath {3{,}946}&\pgfutilensuremath {3{,}952}&\pgfutilensuremath {13.65}&\pgfutilensuremath {5.26}&\pgfutilensuremath {2.49}&\pgfutilensuremath {3.94}&\pgfutilensuremath {1.07}&\pgfutilensuremath {0.89}\\%
rbz43748&\pgfutilensuremath {12.31}&\pgfutilensuremath {75{,}181}&\pgfutilensuremath {67{,}511}&\pgfutilensuremath {106{,}454}&\pgfutilensuremath {106{,}775}&\pgfutilensuremath {405.81}&\pgfutilensuremath {133.68}&\pgfutilensuremath {83.08}&\pgfutilensuremath {66.30}&\pgfutilensuremath {96.36}&\pgfutilensuremath {26.32}\\%
sra104815&\pgfutilensuremath {11.79}&\pgfutilensuremath {168{,}135}&\pgfutilensuremath {150{,}131}&\pgfutilensuremath {255{,}919}&\pgfutilensuremath {257{,}035}&\pgfutilensuremath {1{,}937.60}&\pgfutilensuremath {559.21}&\pgfutilensuremath {191.32}&\pgfutilensuremath {198.95}&\pgfutilensuremath {922.19}&\pgfutilensuremath {65.76}\\%
xia16928&\pgfutilensuremath {12.2}&\pgfutilensuremath {30{,}034}&\pgfutilensuremath {27{,}024}&\pgfutilensuremath {41{,}132}&\pgfutilensuremath {41{,}290}&\pgfutilensuremath {170.36}&\pgfutilensuremath {58.51}&\pgfutilensuremath {31.51}&\pgfutilensuremath {28.97}&\pgfutilensuremath {40.82}&\pgfutilensuremath {10.52}\\%
xib32892&\pgfutilensuremath {12.57}&\pgfutilensuremath {56{,}931}&\pgfutilensuremath {50{,}078}&\pgfutilensuremath {80{,}124}&\pgfutilensuremath {80{,}512}&\pgfutilensuremath {329.57}&\pgfutilensuremath {108.02}&\pgfutilensuremath {62.72}&\pgfutilensuremath {54.01}&\pgfutilensuremath {84.03}&\pgfutilensuremath {20.75}\\%
xit1083&\pgfutilensuremath {10.76}&\pgfutilensuremath {1{,}405}&\pgfutilensuremath {1{,}322}&\pgfutilensuremath {2{,}681}&\pgfutilensuremath {2{,}686}&\pgfutilensuremath {9.20}&\pgfutilensuremath {3.37}&\pgfutilensuremath {1.59}&\pgfutilensuremath {2.88}&\pgfutilensuremath {0.80}&\pgfutilensuremath {0.56}\\%
xmc10150&\pgfutilensuremath {13.77}&\pgfutilensuremath {29{,}050}&\pgfutilensuremath {25{,}112}&\pgfutilensuremath {22{,}906}&\pgfutilensuremath {23{,}001}&\pgfutilensuremath {159.13}&\pgfutilensuremath {35.94}&\pgfutilensuremath {20.61}&\pgfutilensuremath {20.02}&\pgfutilensuremath {75.77}&\pgfutilensuremath {6.77}\\%
xpr2308&\pgfutilensuremath {12.34}&\pgfutilensuremath {5{,}094}&\pgfutilensuremath {4{,}303}&\pgfutilensuremath {5{,}327}&\pgfutilensuremath {5{,}353}&\pgfutilensuremath {27.51}&\pgfutilensuremath {7.81}&\pgfutilensuremath {3.96}&\pgfutilensuremath {5.73}&\pgfutilensuremath {8.77}&\pgfutilensuremath {1.24}\\%
xqc2175&\pgfutilensuremath {11.74}&\pgfutilensuremath {3{,}909}&\pgfutilensuremath {3{,}331}&\pgfutilensuremath {5{,}206}&\pgfutilensuremath {5{,}228}&\pgfutilensuremath {23.40}&\pgfutilensuremath {6.90}&\pgfutilensuremath {3.50}&\pgfutilensuremath {4.95}&\pgfutilensuremath {6.91}&\pgfutilensuremath {1.13}\\%
xqd4966&\pgfutilensuremath {11.02}&\pgfutilensuremath {5{,}965}&\pgfutilensuremath {5{,}412}&\pgfutilensuremath {12{,}554}&\pgfutilensuremath {12{,}578}&\pgfutilensuremath {40.84}&\pgfutilensuremath {18.19}&\pgfutilensuremath {7.55}&\pgfutilensuremath {8.24}&\pgfutilensuremath {4.44}&\pgfutilensuremath {2.42}\\%
xqe3891&\pgfutilensuremath {12.76}&\pgfutilensuremath {8{,}054}&\pgfutilensuremath {6{,}894}&\pgfutilensuremath {9{,}112}&\pgfutilensuremath {9{,}161}&\pgfutilensuremath {44.74}&\pgfutilensuremath {14.24}&\pgfutilensuremath {6.84}&\pgfutilensuremath {9.80}&\pgfutilensuremath {11.83}&\pgfutilensuremath {2.03}\\%
xqf131&\pgfutilensuremath {8.89}&\pgfutilensuremath {154}&\pgfutilensuremath {151}&\pgfutilensuremath {316}&\pgfutilensuremath {316}&\pgfutilensuremath {1.48}&\pgfutilensuremath {0.79}&\pgfutilensuremath {0.16}&\pgfutilensuremath {0.35}&\pgfutilensuremath {0.05}&\pgfutilensuremath {0.13}\\%
xqg237&\pgfutilensuremath {8.99}&\pgfutilensuremath {328}&\pgfutilensuremath {297}&\pgfutilensuremath {569}&\pgfutilensuremath {570}&\pgfutilensuremath {2.23}&\pgfutilensuremath {1.05}&\pgfutilensuremath {0.29}&\pgfutilensuremath {0.53}&\pgfutilensuremath {0.18}&\pgfutilensuremath {0.17}\\%
xql662&\pgfutilensuremath {11.52}&\pgfutilensuremath {1{,}362}&\pgfutilensuremath {1{,}139}&\pgfutilensuremath {1{,}539}&\pgfutilensuremath {1{,}544}&\pgfutilensuremath {8.01}&\pgfutilensuremath {2.56}&\pgfutilensuremath {1.06}&\pgfutilensuremath {1.88}&\pgfutilensuremath {2.11}&\pgfutilensuremath {0.39}\\%
xrb14233&\pgfutilensuremath {12.13}&\pgfutilensuremath {25{,}261}&\pgfutilensuremath {22{,}694}&\pgfutilensuremath {34{,}180}&\pgfutilensuremath {34{,}263}&\pgfutilensuremath {134.89}&\pgfutilensuremath {44.26}&\pgfutilensuremath {25.84}&\pgfutilensuremath {22.67}&\pgfutilensuremath {34.21}&\pgfutilensuremath {7.89}\\%
xrh24104&\pgfutilensuremath {12.7}&\pgfutilensuremath {42{,}058}&\pgfutilensuremath {37{,}583}&\pgfutilensuremath {58{,}217}&\pgfutilensuremath {58{,}436}&\pgfutilensuremath {238.18}&\pgfutilensuremath {83.94}&\pgfutilensuremath {46.79}&\pgfutilensuremath {42.07}&\pgfutilensuremath {50.93}&\pgfutilensuremath {14.42}\\%
xsc6880&\pgfutilensuremath {12.86}&\pgfutilensuremath {11{,}553}&\pgfutilensuremath {9{,}858}&\pgfutilensuremath {16{,}738}&\pgfutilensuremath {16{,}814}&\pgfutilensuremath {68.72}&\pgfutilensuremath {23.64}&\pgfutilensuremath {12.89}&\pgfutilensuremath {12.18}&\pgfutilensuremath {16.44}&\pgfutilensuremath {3.57}\\%
xua3937&\pgfutilensuremath {13.35}&\pgfutilensuremath {10{,}602}&\pgfutilensuremath {8{,}898}&\pgfutilensuremath {9{,}160}&\pgfutilensuremath {9{,}192}&\pgfutilensuremath {71.49}&\pgfutilensuremath {14.24}&\pgfutilensuremath {7.29}&\pgfutilensuremath {11.38}&\pgfutilensuremath {36.12}&\pgfutilensuremath {2.46}\\%
xva2993&\pgfutilensuremath {11.8}&\pgfutilensuremath {4{,}584}&\pgfutilensuremath {3{,}969}&\pgfutilensuremath {7{,}354}&\pgfutilensuremath {7{,}382}&\pgfutilensuremath {31.16}&\pgfutilensuremath {10.35}&\pgfutilensuremath {4.90}&\pgfutilensuremath {6.62}&\pgfutilensuremath {7.76}&\pgfutilensuremath {1.52}\\%
xvb13584&\pgfutilensuremath {12.28}&\pgfutilensuremath {19{,}334}&\pgfutilensuremath {17{,}576}&\pgfutilensuremath {33{,}872}&\pgfutilensuremath {33{,}940}&\pgfutilensuremath {118.64}&\pgfutilensuremath {44.98}&\pgfutilensuremath {24.82}&\pgfutilensuremath {22.13}&\pgfutilensuremath {18.47}&\pgfutilensuremath {8.22}\\%
\end {longtabu}\egroup %
\endgroup %

\end{landscape}

\section{MWT of Mona Lisa}
\label{app::mona_lisa}

The Mona Lisa TSP instance contains 100,000 points and was created by Robert Bosch in February 2009. It yields a representation of Leonardo da Vinci's Mona Lisa as a continuous-line drawing. 
The optimal TSP tour is currently unknown, however, a minimum weight triangulation can be obtained with the LMT-heuristic in less than a second. The MWT of Mona Lisa is shown in Figure \ref{fig::mona_lisa}.

\begin{figure}[h]
\includegraphics[width=\columnwidth]{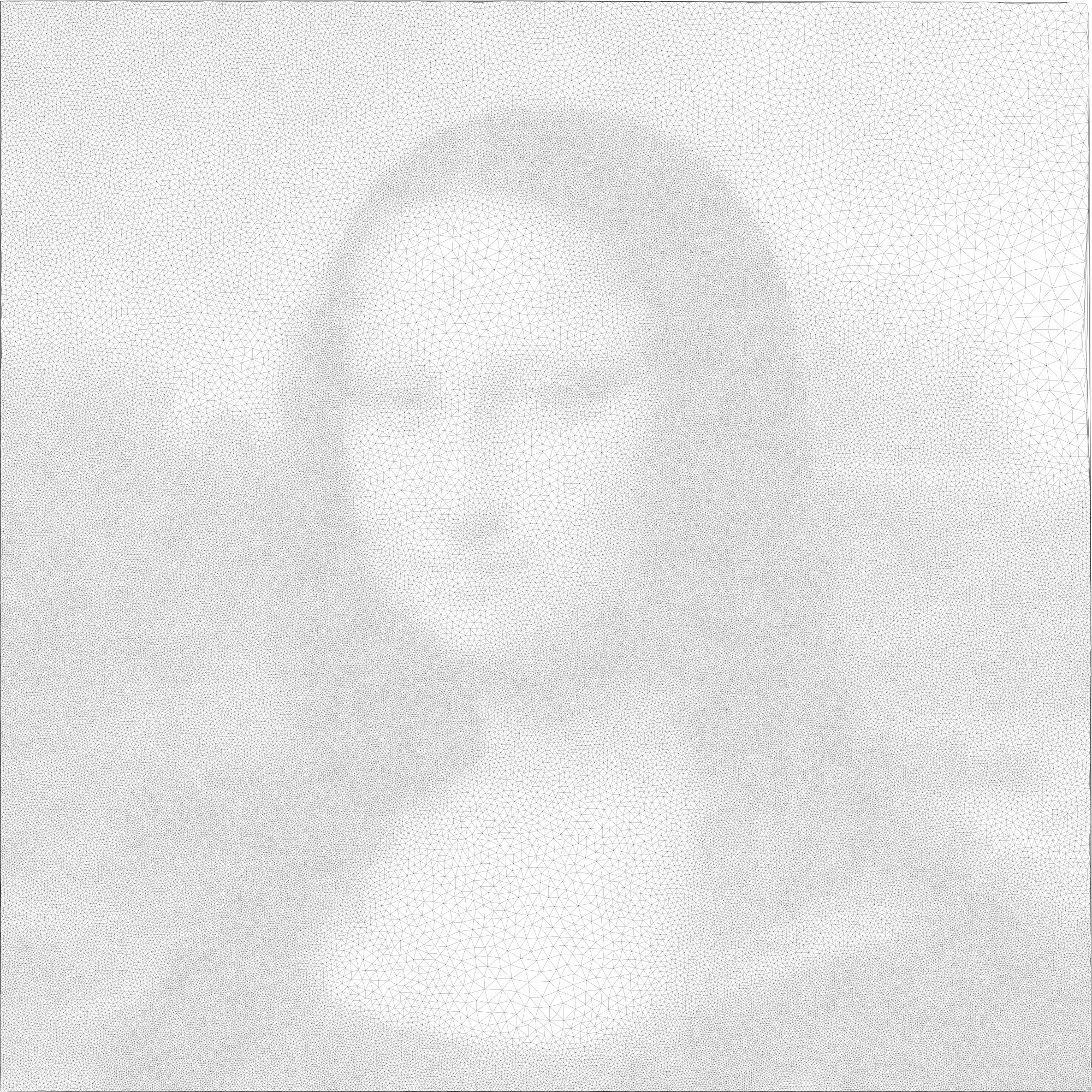}
\caption{MWT of Mona Lisa}
\label{fig::mona_lisa}
\end{figure}

Her weight, you ask? She won't tell $\dots$

\newpage
\section{Pseudo-Angles}
\label{app::pseudo_angles}

In terms of performance, the pseudo-angle computation is 30 to 40 times faster than \texttt{atan2}. (Code compiled with gcc 5.4.0 and tested on an Intel i5-3210M, Intel i7-4770 and Intel i7-6770K). 
The impact on the overall runtime of the diamond test is nearly a factor of 2.

\begin{figure}[h]
	\begin{lstlisting}[language=C++]
	double pseudo_angle(double x, double y) {
	return std::copysign(
	1.0 - x / (std::fabs(x) + std::fabs(y)), y);
	}
	\end{lstlisting}
	\caption{Pseudo angle function in C++.}
	\label{code::pseudo_angle}
\end{figure}

The above function has the same general structure as \texttt{atan2}: they monotonically increase in the intervals $[0, \pi]$, $(\pi, 2\pi)$ and have their discontinuity at $\pi$, where they jump from positive to negative values. See Figure \ref{plot::pseudo_angle} for a plot.

\begin{figure}[h]
	\includegraphics[width=\columnwidth]{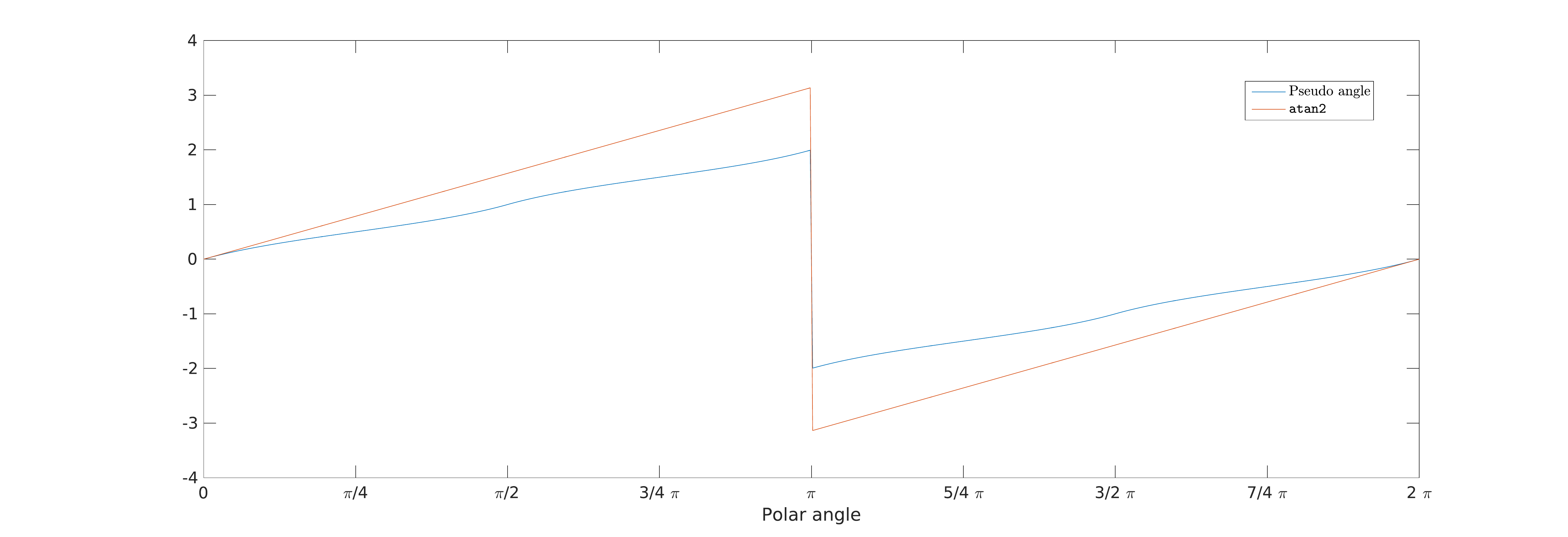}
	\caption{\texttt{atan2} versus pseudo-angle.}
	\label{plot::pseudo_angle}
\end{figure}

Beware that pseudo-angles are not rotationally invariant, i.e., a constant angle difference between two arbitrary vectors is not constant when measured with pseudo-angles. For example, a real polar angle of $\pi/4.6$ corresponds to a pseudo-angle of $\approx 0.4486$, but two arbitrary vectors that are $\pi/4.6$ apart can differ by as much as $\approx 0.525$ in their pseudo-angles.

\end{document}